\documentclass[11pt,english]{article}
\pdfoutput=1
\usepackage[T1]{fontenc}
\usepackage[latin9]{inputenc}
\usepackage{geometry}
\usepackage{verbatim}
\usepackage{amsmath}
\usepackage{amssymb}
\usepackage{graphicx}

\usepackage{url}

\usepackage{amsthm}
\usepackage{amsmath}
\usepackage{amssymb}

\usepackage{times}

\makeatletter

\usepackage{amsthm}\usepackage{enumerate}

\theoremstyle{plain}
\newtheorem{thm}{\protect\theoremname}
  \theoremstyle{plain}
  
  \theoremstyle{plain}
  \newtheorem{lem}[thm]{\protect\lemmaname}
  \theoremstyle{plain}
  \newtheorem{prop}[thm]{\protect\propositionname}
  \theoremstyle{definition}
  \newtheorem{defn}[thm]{\protect\definitionname}

\usepackage[dvipsnames]{xcolor}\newcounter{note}[section] 
\newcommand{\aanote}[1]{}
\newcommand{\awnote}[1]{}
\newcommand{\aw}[1]{#1}
\newcommand{\aad}[1]{#1}

\usepackage{babel}
\providecommand{\definitionname}{Definition}
  \providecommand{\lemmaname}{Lemma}
  \providecommand{\propositionname}{Proposition}
\providecommand{\theoremname}{Theorem}
\providecommand{\corollaryname}{Corollary}

\usepackage{babel}

\makeatother

\usepackage{babel}
\begin{document}
\global\long\def\L{\mathcal{L}}
 \global\long\def\eps{\epsilon}
 \global\long\def\P{\mathcal{P}}
 \global\long\def\Le{\L_{\mathrm{ext}}}
 \global\long\def\F{\mathcal{F}}
 \global\long\def\E{\mathcal{E}}
 \global\long\def\D{\mathcal{D}}
 \global\long\def\C{\mathcal{C}}
 \global\long\def\Q{\mathcal{Q}}
 \global\long\def\poly{\mathrm{poly}}
\global\long\def\T{\mathcal{T}}
\global\long\def\wb{\bar{w}}
\global\long\def\own{\mathrm{own}}
\global\long\def\opt{\mathrm{opt}}

\title{A QPTAS for Maximum Weight Independent Set of Polygons \\
with Polylogarithmically Many Vertices}
\date{$\,$}
\author{
Anna Adamaszek\footnote{Max-Planck-Institut f\"ur Informatik, Saarbr\"ucken, Germany,
       \texttt{\{anna,awiese\}@mpi-inf.mpg.de}}
\and
Andreas Wiese\footnotemark[1]
}

\maketitle
\thispagestyle{empty}
\begin{abstract}
The Maximum Weight Independent Set of Polygons (MWISP) problem is a fundamental problem in computational geometry.
Given a set of weighted polygons in the 
two-dimensional plane, the goal is to find a set of pairwise non-overlapping polygons with maximum
total weight.
Due to its wide range of applications and connections to other problems, the MWISP problem and its special cases have been extensively studied both in the approximation algorithms and the computational geometry community.
Despite a lot of research, its general case is not well-understood yet.
Currently the best known polynomial time algorithm achieves an
approximation ratio of $n^{\epsilon}$
[Fox and Pach, SODA 2011], and it is not even clear whether the problem is $\mathsf{APX}$-hard. We present a $(1+\epsilon)$-approximation
algorithm, assuming that each polygon in the input
has at most a polylogarithmic number of vertices. Our algorithm has quasi-polynomial
running time, i.e., it runs in time $2^{\mathrm{poly}(\log n,1/\eps)}$.
In particular, our result implies that for this setting the problem is \emph{not} $\mathsf{APX}$-hard, unless
$\mathsf{NP}\subseteq\mathsf{DTIME}(2^{\mathrm{poly}(\log n)})$. 

We use a recently introduced framework for approximating maximum weight independent
set in geometric intersection graphs. The framework
has been used to construct a QPTAS in the much simpler 
case of axis-parallel rectangles. We extend it in two ways, to adapt it to
our much more general setting. First, we show that its technical core can be reduced to
the case when all input polygons are triangles. Secondly, 
we replace its key technical ingredient which is 
a method to partition the plane using only few edges such that the
objects stemming from the optimal solution are evenly distributed
among the resulting faces and each object is intersected only a few times.
Our new procedure for this task is not even more complex than the original one
and, importantly, it can handle the arising difficulties due
to the arbitrary angles of the input polygons. Note that already this
obstacle makes the known analysis for the above framework fail. Also, in general
it is not 
well understood how to handle this difficulty by efficient approximation
algorithms.
\end{abstract}
\newpage{}

\setcounter{page}{1}

\section{Introduction}

One of the classical problems in combinatorial optimization is the
\textsc{Maximum Independent Set} problem. Given a graph, the goal
is to select a subset of the vertices such that no two selected vertices
are connected by an edge. The objective is to maximize the cardinality
or the total weight of the selected vertices. While this is a fundamental
problem, in its full generality it is essentially computationally
intractable for efficient algorithms: the problem is $\mathsf{NP}$-hard
to approximate within a factor of $O(n^{1-\eps})$ for any $\eps>0$~\cite{zuck07}.

However, when restricting to special graph classes, it is possible
to obtain polynomial time algorithms with much better approximation
guarantees. An important class are intersection graphs of two-dimensional
objects in the plane, i.e., 
where
each vertex in the graph corresponds
to a shape in the plane and two vertices are connected by an edge
if and only if their corresponding shapes 
overlap. Without loss
of generality, we can assume that the given shapes are (not necessarily
convex) polygons. In this paper, we study the resulting Maximum Weight
Independent Set of Polygons (MWISP) problem. It arises in many settings like map labeling~\cite{AKS1998,sack1999handbook,verweij1999optimisation},
chip manufacturing~\cite{Hochbaum1985}, cellular networks~\cite{clark1990unit}, 
unsplittable flow~\cite{anagnostopoulos2013constant,BSW11},
or data mining~\cite{fukuda2001data,KMP1998,lent1997clustering}.

Computationally, the MWISP problem is not well-understood. When interpreting
the problem equivalently as finding a set of non-intersecting curves
in the plane, the best known polynomial time algorithm achieves an
approximation ratio of $n^{\eps}$~\cite{FoxPach2011}.
On the other hand, the best known hardness result 
for the geometric setting is strong $\mathsf{NP}$-hardness.
This holds already if all 
objects
are straight 
line segments
with at most two different angles~\cite{kratochvil1990independent}.
Thus, there is a very large gap between the known approximation and hardness results for the
MWISP problem.
In this paper we completely close this gap under the assumption that the polygons do not have too many vertices (at most $(\log n)^{O(1)}$ many) 
and when allowing quasi-polynomial running time. 

\subsection{Our Contribution}

We present a $(1+\eps)$-approximation algorithm for
the maximum weight independent set of polygons problem, running in quasipolynomial
time 
$2^{\mathrm{poly}(\log n,1/\eps)}$,
where we assume that each polygon has at most a polylogarithmic
number of vertices. As mentioned above, in comparison the best known
polynomial time algorithm for this problem achieves an approximation
ratio of $n^{\eps}$. 
To the best of our knowledge, no better approximation
algorithms are known that run in quasi-polynomial time. 

We use a recently introduced framework for approximating maximum weight independent set in geometric intersection graphs~\cite{AW2013}. The framework has been introduced to construct a QPTAS in the much simpler case where all input objects are axis-parallel rectangles. 
However, for arbitrary polygons the problem is much more difficult and the techniques from~\cite{AW2013} alone are by far not sufficient. 
There are two major obstacles we have to overcome: handling arbitrary angles and dealing with more complex polygons, i.e., polygons with a super-constant number of vertices. 
\aw{
In fact, it is not well understood how to handle these obstacles by efficient algorithms.
For instance, if all polygons are axis-parallel rectangles,}
many algorithms are known that achieve an $O(\log n)$-approximation ratio or better~\cite{AKS1998,berman2001improved,CC2009,ChanHarPeled2012,KMP1998,N2000}. However, already for straight line segments with arbitrary angles, 
that are a special case of
non-axis-parallel rectangles, the best known polynomial time algorithm is the $n^{\eps}$-approximation mentioned above~\cite{FoxPach2011}. 
Generalizations of axis-parallel rectangles to polygons with more than just four axis-parallel edges are not much better understood either.
In particular, the LP approach (which has been extensively exploited in the case of axis-parallel rectangles) does not work in a more general setting.
Already for non-axis-parallel rectangles and axis-parallel L-shapes (i.e., polygons with six vertices each) the integrality gap of the LP can be as large as $\Omega(n)$. We show how to overcome 
the two major obstacles discussed above
and thus in particular contribute towards better algorithmic understanding of them.


Our algorithm is a dynamic
program parametrized by an integer $k$, where for our result we choose
$k=\poly(\log n,1/\eps)$. Its DP-table has an entry for each polygon
in the plane (not necessarily part of the input) with at most $k$
edges. In order to bound the size of the DP-table, we require the
corners of these polygons to be from some suitable polynomial size
set.
Each DP-cell represents the subproblem given by all input
polygons that are contained in the polygon associated with the DP-cell.
When computing the value for the cell, the dynamic program tries all
combinations to partition its corresponding polygon into up to $k$
smaller polygons with at most $k$ edges each, and selects the most
profitable combination according to the precomputed values for the
DP-cells of the smaller polygons.

For bounding the approximation ratio of our DP, 
we first show that we can reduce the general setting to the special
case where all polygons are triangles. In fact, it is even true that
if for some parameter $k$ our algorithm is a $(1-\eps)$-approximation for triangles,
it yields a $(1-K\cdot\eps)$-approximation for arbitrary $K$-gons.
Observe that the latter implication cannot be applied to arbitrary
algorithms that work well on triangles. 
Secondly, we show that for any feasible (and in particular the optimal) solution
for a problem instance where all polygons are triangles
there is a balanced cut with few edges that intersects only 
triangles with small total weight. Applying such a cut recursively for $O(\log n/\epsilon)$ levels
eventually subdivides the input plane into pieces such that each of them contains only a single 
triangle of the optimal solution. Note that the DP has to guess these cuts, so it is
very important that they have small complexity.

For showing that a balanced cut always exists, our technical core is to show
that there is a partition of the plane into faces,
using only a bounded number of edges, such that  
for each face of the partition the total weight of 
triangles from the optimal solution touching the face
is relatively small and each triangle
is intersected only a bounded number of times by edges of the faces. 
The techniques from~\cite{AW2013} for axis-parallel rectangles entirely fail here.
One key ingredient of the partition for 
rectangles is a procedure that 
stretches the input area non-uniformly along the two dimensions, ensuring that the rectangles 
contained in 
any thin vertical or horizontal stripe
of the input have small total weight. The 
partition is then constructed so that each face is contained in a bounded number of such stripes. 
For triangles with arbitrary angles, 
a stretching procedure preserving the combinatorial structure of the input
would turn them again into polygons with many bends. 
Even worse, when trying to adapt the construction from~\cite{AW2013}, the arbitrary 
angles of the triangles cause that one cannot guarantee that
the faces are contained in a bounded number of thin horizontal or vertical stripes, or even in 
thin stripes with a small number of suitably chosen directions.

Therefore, we provide a completely new construction 
of the partition which is much more general than the specialized version for axis-parallel rectangles,
and that is not even more complicated than the latter.
Altogether, this yields a QPTAS for the MWISP problem for arbitrary $K$-gons with $K=\poly(\log n)$. 
We note that our algorithm---the geometric dynamic program---might well be the right
algorithmic technique for better \emph{polynomial time} approximation algorithms
for MWISP, when parametrized by $k=O(1).$

\subsection{Related Work}

There is a lot of research on the maximum independent set problem
for geometric shapes in the plane. 
For the case when all shapes are unit disks (i.e.,
the resulting intersection graph is a unit disk graph),
polynomial time approximation schemes have
been presented by Hunt III et al.~\cite{PTAS_several_HuntIII}, 
and by Hochbaum and Maass~\cite{Hochbaum1985}. 
Subsequently, Nieberg, Hurink, and Kern~\cite{PTAS-indset} presented a PTAS that does 
not even need the geometric representation of the graph.
Observe here that it is $\mathsf{NP}$-hard to decide whether a given
graph is a unit disk graph~\cite{breu98_UDG_recognition}. Using
the geometric embedding, Erlebach, Jansen, and Seidel presented a 
PTAS for disks with arbitrary diameters, which generalizes to arbitrary fat objects~\cite{EJS2005}.

When going beyond \aad{the} fat objects, the problem
becomes much less understood.
Even for axis-parallel rectangles the best known polynomial time algorithms
are $O(\log n/\log\log n)$-approximation for weighted case due to Chan and Har-Peled~\cite{ChanHarPeled2012}, 
and a $O(\log\log n)$-approximation for unweighted case by Chalermsook and Chuzhoy~\cite{CC2009}.
Prior to the latter results, many $O(\log n)$-approximation algorithms
have been found~\cite{AKS1998,berman2001improved,KMP1998,N2000}. Very
recently, a QPTAS has been presented~\cite{AW2013}.

For the general case, the best known result is a $n^{\eps}$-approximation
algorithm for collections of curves due to Fox and Pach~\cite{FoxPach2011}, which assumes that any two curves
intersect only $\ell=O(1)$ times, i.e., that they are $\ell$-intersecting.
Prior to this, Agarwal and Mustafa presented an algorithm for the special case of 
straight line segments that computes an
independent set of size $(OPT/\log(2n/OPT))^{1/2}$~\cite{agarwal2006independent},
which yields a worst case approximation ratio of $n^{1/2+o(1)}$.
Note that it is $\mathsf{NP}$-hard to decide whether a given graph
can be represented as an intersection graph of a collection of curves
(i.e., is a string graph)~\cite{schaefer2003recognizing}.



\subsection{Problem Definition and Notation}

We are given a set of $n$ polygons $\P$ and an integer $K$. Each
polygon $P_{i}\in\P$ is a simple polygon specified by a sequence
of vertices $p_{i,1},...,p_{i,\ell}$, for some $\ell\le K$, each
vertex $p_{i,j}=(x_{i,j},y_{i,j})$ having integer coordinates $x_{i,j},y_{i,j}\in\mathbb{N}$.
Each polygon is considered as an open set. 
Also, each polygon $P_{i}$
has a weight $w(P_{i})>0$ associated with it. The goal is to find
a maximum weight independent set of polygons, i.e., a set $\P'\subseteq\P$
such that for all pairs of polygons $P,P'\in \P'$ with $P\ne P'$ it holds that $P\cap P'=\emptyset$. We
will show that for $K=\textrm{poly}(\log n)$ this problem admits a
QPTAS. Without loss of generality we assume that 
no three vertices of the input polygons
lie on a straight line
(i.e., the input points are in general position, see Appendix~\ref{apx:general-position} for details).

For any two points $p,p'$ we define $L[p,p']$ to be the straight
line segment connecting $p$ and $p'$. For two point sets $A,A'$ (which
could be lines, polygons, etc.) we say that $A$ \emph{touches} $A'$
if $A\cap A'\ne\emptyset$. We say that $A$ \emph{crosses $A'$ if
$A'\setminus A$ }has at least two connected components. 
For a given 
problem
instance (that will always be clear from the context)
we assume w.l.o.g.~that there is an integer $N$---not necessarily polynomially bounded---such that for
all points $p_{i,j}$ it holds that $p_{i,j}\in I:=[0,N]\times[0,N]$.
For a set of polygons $\P' \subseteq \P$ we define $w(\P'):=\sum_{P\in \P'}w(P)$.
Due to space constraints, the proof of the next lemma, as well as some of the other proofs, 
are in the appendix.


\begin{lem} \label{lem:bound-weights} By losing at most a factor of 
$1+O(\eps)$
we can assume that $1\le w(P)\le n/\eps$ for each polygon $P\in\P$. \end{lem}


\section{Dynamic Program}

We describe our dynamic program, parametrized
by an integer parameter $k$, for approximating the MWISP problem. The dynamic program is a generalization of the
DP presented in~\cite{AW2013} for approximating the maximum weight independent
set of axis-parallel rectangles (rather than general polygons).

Intuitively, the DP has an entry in its DP-table for each polygon
with at most $k$ edges, 
possibly with holes, within the input square. To make the size
of the table polynomially bounded, the corners of the polygons are chosen from
a suitable, polynomially bounded set of points. The entry for each
polygon $Q$ represents a good approximation on the optimal solution
for the subproblem given by the input polygons from $\P$ that are
contained in $Q$. In order to compute an entry for a polygon $Q$,
the dynamic program tries all possibilities to subdivide $Q$ into
at most $k$ smaller polygons with at most $k$ edges each. 
In a leaf subproblem~$Q$, at most one polygon from $\P$ contained in $Q$ is chosen. 

To define the DP formally, we first define the set of possible corners
of the polygons in the DP-table, which we call the \emph{DP-points}.
For each polygon $P \in \P$ we fix a triangulation of $P$ and we denote by $\T(P)$ the resulting set of triangles.
Denote by $\E_T$ the set of all boundary edges of the triangles $\T(P)$ for all polygons $P \in \P$. (Notice that if all input polygons are triangles, $\E_T$ is the set of their boundary edges.) The set of \emph{basic DP-points} contains the four corners of the input square $I$ and each intersection of a vertical line $\{(x,y)|y\in\mathbb{R}\}$ such that $x=x_{i,j}$ for some polygon $P_{i}\in\P$ with a line $L\in\E_T$ or with the horizontal boundary edge of the input square, see Figure~\ref{fig:DP}.
The set of \emph{additional DP-points} contains 
all intersection points of pairs of line segments whose endpoints are basic DP-points, which are contained in $I$.
When we use the notion of DP-points, we refer to the set of all, basic \emph{and} additional DP-points.
Note that all
corners of polygons from $\P$ are basic DP-points and the number of all DP-points is upper
bounded by $O((nK)^{4})$.
Denote by $\Q$ the set of all polygons with at most $k$ edges and possibly with holes, whose corners are DP-points.
In the DP-table, we have an entry for each
polygon $Q \in \Q$. To distinguish polygons from the input $\P$ and polygons
of the DP-table, 
we will always denote the latter by~$Q$, $Q'$, etc.

\begin{prop} \label{prop:DP-cells} The number of DP-cells is at
most $(nK)^{O(k)}$. \end{prop}

For each polygon $Q \in \Q$ the dynamic program computes
a set $sol(Q)$ of pairwise non-touching polygons from the set $\P_{Q} := \{P \in \P: P \subseteq Q\}$.
The set $sol(Q)$ is computed as follows. If $|\P_{Q}| \le 1$, we set $sol(P):=\P_{Q}$ and terminate. 
Otherwise, we enumerate all possible partitions of $Q$ into at most $k$
polygons from $\Q$, see Figure~\ref{fig:DP}.  By Proposition \ref{prop:DP-cells} we have
$|\Q| \le (nK)^{O(k)}$, so the number of potential partitions we need to consider is upper bounded
by ${(nK)^{O(k)} \choose k}=(nK)^{O(k^{2})}$. 
Also, for any set of at most $k$ polygons in $\Q$ we can check efficiently whether they form a partition of some larger polygon in $\Q$ (e.g., we can use as a skeleton the subdivision of the plane obtained by drawing all line segments between any pair of DP-points).
Let $\{ Q_{1},...,Q_{k'} \}$, with $k' \le k$, denote the partition that maximizes $\sum_{i=1}^{k'}w(sol(Q_{i}))$, and let $sol'(Q):=\cup_{i=1}^{k'}sol(Q_{i})$.
We set $sol(Q) := sol'(Q)$ if $w(sol'(Q))> \max_{P\in\P_Q}w(P)$,
and otherwise $sol(Q):=\{P_{\max}\}$ for a polygon $P_{\max} \in \P_Q$ with maximum profit in $\P_Q$.
At the end, the algorithm outputs the value in the DP-cell which corresponds
to the polygon containing the entire input region.
We call this algorithm \emph{GEO-DP}, like its specialization for axis-parallel rectangles given in~\cite{AW2013}.

\begin{figure}[t]
\begin{centering}
\includegraphics[scale=0.50]{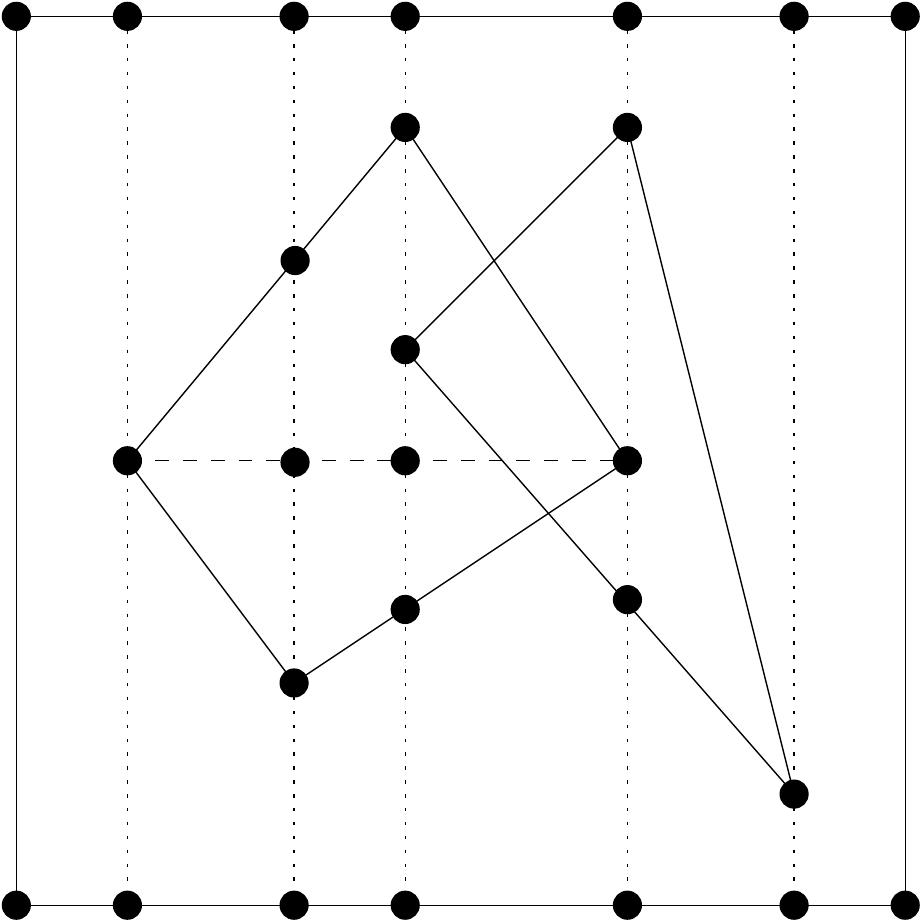}~~~~~~~~~~~~~\includegraphics[scale=0.55]{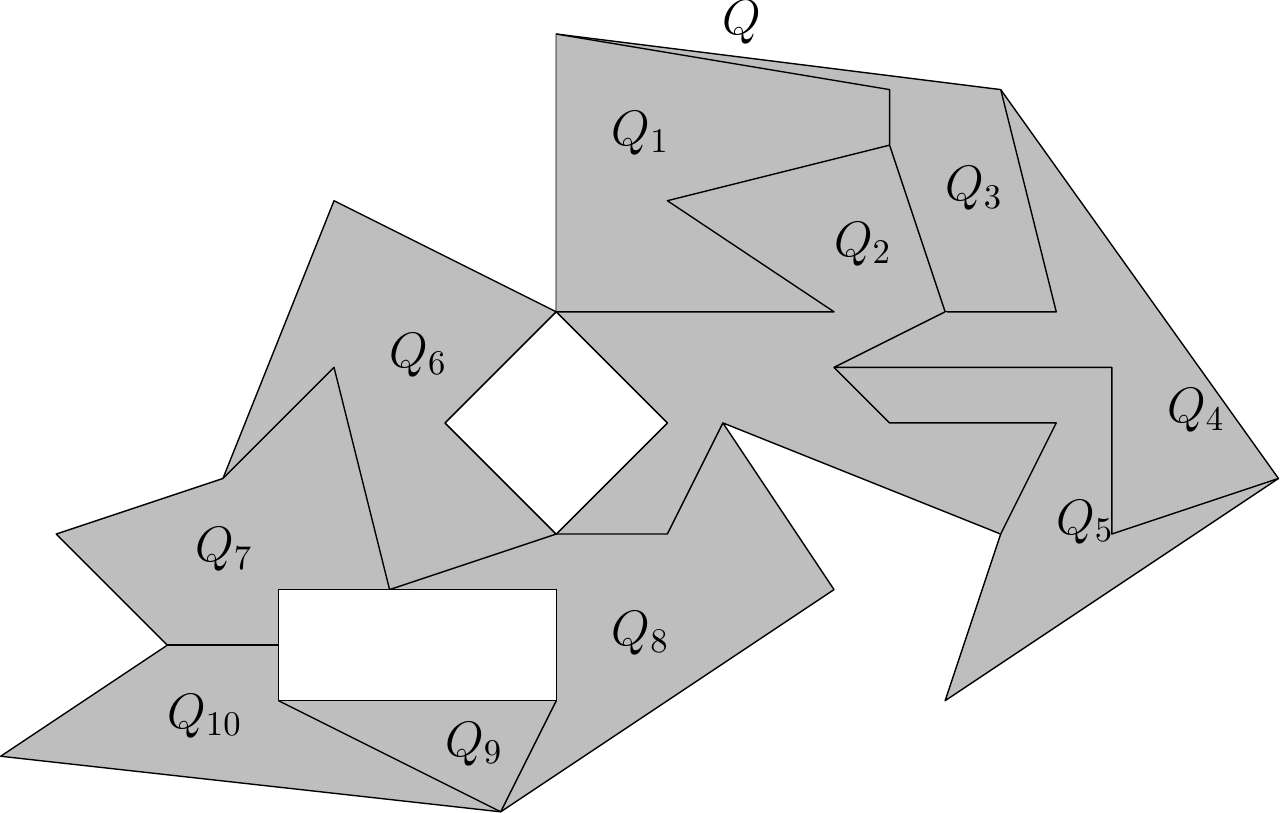}
\par\end{centering}

\caption{\label{fig:DP}Left: 
\emph{Basic }DP-points are denoted by fat dots. The dashed line indicates the triangulation of the left polygon.
Right: The dynamic
program GEO-DP tries all possibilities to subdivide the polygon $Q$
into at most $k$ smaller polygons with at most $k$ edges each. One
partition is shown in the figure.}

\end{figure}

\begin{prop} \label{prop:running-time-GEO-DP} 
When parametrized
by $k$, the running time of GEO-DP is upper bounded by $(nK)^{O(k^{2})}$.
\end{prop}

\subsection{Approximation Ratio}

The core of our reasoning is showing that for any set of pairwise non-touching
polygons there is a balanced cheap cut.

\begin{defn} Let $\ell\in\mathbb{N}$ and $\alpha\in\mathbb{R}$
with $0<\alpha<1$. Let $\bar{\P}$ be a set of pairwise non-touching
polygons. A polygon $\Gamma$ is a \emph{balanced $\alpha$-cheap
$\ell$-cut} if: 
\begin{itemize}
\item $\Gamma$ has at most $\ell$ edges, 
\item for the set of all polygons $\P'\subseteq\bar{\P}$ intersecting the
boundary of $\Gamma$ we have $w(\P')\le\alpha\cdot w(\bar{\P})$, 
\item for the set of all polygons $\P_{\mathrm{in}}\subseteq\bar{\P}$ contained
in $\Gamma$ it holds that $w(\P_{\mathrm{in}})\le2/3\cdot w(\bar{\P})$,
and 
\item for the set of all polygons $\P_{\mathrm{out}}\subseteq\bar{\P}$
contained in the complement of $\Gamma$, i.e., in $I\setminus\Gamma$,
it holds that $w(\P_{\mathrm{out}})\le2/3\cdot w(\bar{\P})$. 
\end{itemize}
\end{defn}

The next lemma shows that if for any set of polygons there is a good
enough balanced cheap cut, then GEO-DP has a good approximation ratio. 

\begin{lem} \label{lem:good-cut-suffices} Let $\eps>0$ and choose
$\alpha:=\frac{\eps}{\log (n/\eps)}$. Let $\ell\ge4$ be a value such that
for any set $\bar{\P}$ of at most $n$ pairwise non-touching polygons
with at most $K$ vertices each there exists a balanced $\alpha$-cheap
$\ell$-cut whose vertices are basic DP-points, or there is a polygon
$P\in\bar{\P}$ such that $w(P)\ge1/3\cdot w(\bar{\P})$. Then GEO-DP
has an approximation ratio of at most $1+O(\eps)$ when parametrized
by $k=O(\ell^{2}\cdot \log^{2}(n/\eps))$. \end{lem}

\begin{proof}[Proof sketch.]Let $\P_{\opt}\subseteq\P$ denote the
optimal solution
for a problem instance $\P$. We start with the polygon $Q_{0}:=\{[0,N]\times[0,N]\}$.
If there is a polygon $P\in\P_{\opt}$ with $w(P)\ge1/3\cdot w(\P_{\opt})$
then GEO-DP guesses the partition $Q_{1}:=Q_{0}\setminus P$ and $Q_{2}:=P$
and recurses on $Q_{1}$. Note that $w(Q_{1})\le2/3\cdot w(\P_{\opt})$.
Otherwise, GEO-DP guesses a balanced $\alpha$-cheap $\ell$-cut $\Gamma$ for $\P_{\opt}$
and recurses on $Q'_{1}:=\Gamma$ and $Q'_{2}:=Q_{0}\setminus\Gamma$.
Note that $w(Q'_{1})\le2/3\cdot w(\P_{\opt})$ and $w(Q'_{2})\le2/3\cdot w(\P_{\opt})$,
and the total weight of intersected polygons from $\P_{\opt}$ is upper
bounded by $\alpha \cdot w(\P_{\opt})=\frac{\eps}{\log(n/\eps)} \cdot w(\P_{\opt})$. We continue recursively
for $O(\log(n/\eps))$ levels until we obtain subproblems containing
at most one polygon from $\P_{\opt}$ each (recall that $w(\P_{\opt})\le n^{2}/\eps$,
see Lemma~\ref{lem:bound-weights}). As in each level we intersect,
i.e., lose, polygons of total weight at most $\alpha \cdot w(\P_{\opt})$,
the total approximation ratio is $1/(1-\alpha)^{O(\log(n/\eps))}\le 1+O(\eps)$.

Each polygon $Q_i$ inducing a subproblem in each level can be expressed
as an intersection of at most $O(\log(n/\eps))$ polygons with at
most $\ell$ edges each, all corners of those being basic DP-points.
Hence, the resulting polygon has at most $k$ edges and all its corners
are basic or additional DP-points. Also, when applying a cut $\Gamma$
on such a subproblem $Q_{i}$, the number of resulting connected components of $\Gamma$ and $Q_i \setminus \Gamma$ is upper bounded
by $k$. Hence, when executed with parameter $k$, eventually GEO-DP
will guess the sequence of cuts described above and thus compute a
$(1+O(\eps))$-approximative solution.\end{proof}


\section{Finding a Balanced Cheap Cut}

In this section we show that for any set of pairwise non-touching polygons with at most $K=\textrm{poly}(\log n)$ vertices each we can always find a balanced cheap cut with good parameters. Together with Lemma \ref{lem:good-cut-suffices} this will prove that the algorithm GEO-DP parametrized by a suitable parameter $k$ is a QPTAS, assuming that the input consists of polygons with at most $K$ vertices each.

Our reasoning has three steps. First, we show that we can reduce the problem of finding the desired cut to the special case when all polygons are triangles. In the second step we construct a partition of the plane that uses only a bounded number of edges and which has some 
useful properties. Finally, we apply a a planar graph separator theorem from \cite{Arora1998} to the partition, which then yields the cut.

\subsection{Reduction to Triangles}

We show that up to a factor of $K$ in the cost of the cut and its
complexity, it suffices to find a good cut for the case that all polygons
are triangles.

\begin{lem} \label{lem:reduction-to-triangles}
Assume that for any $\delta>0$
and for any set of pairwise non-touching triangles $\T$ 
such that $\max_{T\in\T}w(T)<w(\T)/3$
there is a balanced $\delta$-cheap $f(\delta)$-cut
with vertices at basic DP-points
for some function $f$. Then for any $\delta>0$ and for any set of pairwise non-touching
polygons $\P$ with at most $K$ edges each such that $\max_{P\in\P}w(P)<w(\P)/3$
there is a balanced $K \cdot \delta$-cheap
$K\cdot f(\delta)$-cut
with vertices at basic DP-points.
\end{lem}

\begin{proof} Consider a set of pairwise non-touching polygons
$\P$ with at most $K$ edges each and let $\delta>0$. We triangulate
each polygon $P\in\P$ into at most $K-2$ triangles $\T(P)$,
using the same triangulation as used for choosing the basic DP-points,
and distribute the weight of $P$ equally among the triangles in $\T(P)$.
Let $\T:=\cup_{P\in\P}\T(P)$ denote the resulting set of pairwise
non-touching triangles. We have that $\max_{T\in\T}w(T)\le\max_{P\in\P}w(P)<w(\P)/3=w(\T)/3$.
By assumption, there is a balanced $\delta$-cheap $f(\delta)$-cut
$\Gamma$ for $\T$
with vertices at basic DP-points. Notice that the sets of basic DP-points for $\P$ and $\T$ are the same. We will transform $\Gamma$ to a $K \cdot \delta$-cheap
$K\cdot f(\delta)$-cut $\Gamma'$ for $\P$.

An edge $e$ of $\Gamma$ can cross a polygon $P\in\P$ in two ways.
The first case is that $e$ touches at least one triangle $T\in\T(P)$.
Since $\Gamma$ is $\delta$-cheap, the total weight of triangles 
touched by edges of $\Gamma$ is
upper bounded by $\delta\cdot w(\T)$.
Therefore, the total weight of polygons $P\in \P$ such that at least one 
triangle of $\T(P)$ is touched by an edge of $\Gamma$
is upper bounded by $\delta(K-2)\cdot w(\T)\le\delta K\cdot w(\P)$.
The other case is that $e$ crosses~$P$, but it does not cross any
triangle from $\T(P)$ (see Figure~\ref{fig:circumvention-of-polygon} in the appendix). Since we assumed that no three vertices of the
input polygons lie on a straight line (and edges of $\Gamma$ are
straight lines), each edge $e$ of $\Gamma$ can cross at most one
polygon $P(e)\in\P$ in this way. We fix this by replacing $e$ by
at most $K$ edges such that $P(e)$ is circumvented and still for
each side of the cut we have the property that the weight of contained
polygons is at most $2/3\cdot w(\P)$. 
Note that here we use that $w(P(e))<w(\P)/3$. 
All newly introduced vertices of the cut $\Gamma'$ are vertices of input polygons $\P$, and so they are basic DP-points.
We perform this operation
for each edge of $\Gamma$. Denote by $\Gamma'$ the resulting cut
for $\P$. By construction, an edge $e'$ of $\Gamma'$ can cross
a polygon $P$ only if there is an edge $e$ of $\Gamma$ that crosses
a triangle $T\in\T(P)$. Hence, $\Gamma'$ is a balanced $K \cdot \delta$-cheap
$K\cdot f(\delta)$-cut,
and all vertices of $\Gamma'$ are basic DP-points.
\end{proof}

In fact, with similar arguments as in the above lemma one can show
that if GEO-DP parametrized by $k$ yields a $(1-\eps)^{-1}$-approximation
algorithm for triangles, then GEO-DP parametrized by some parameter
in $\poly(k,K)$ is a $(1-K\cdot\eps)^{-1}$-approximation algorithm
for polygons with at most $K$ edges. Similarly as above, the idea
is to triangulate each polygon in the input, consider an execution
of GEO-DP on this triangulation, and circumvent polygons $P$ that
are touched without any of their triangles $\T(P)$ being touched.
Hence, an alternative way of showing that GEO-DP is a QPTAS for MWISP if $K\le (\log n)^{O(1)}$
is to prove that for triangles it is a $(1+\epsilon)$-approximation with running time 
$n^{\poly(\log n,1/\epsilon)}$, and then apply the above reasoning. 
Note that the actual dependence of $\epsilon$ on the running time
is crucial for this line of reasoning.

Following the (other) approach suggested by Lemma~\ref{lem:reduction-to-triangles}, in the remainder of this section we prove that for
any set of triangles $\T$ and for any $\delta > 0$ there is a balanced $O(\delta)$-cheap
$\left(\frac{1}{\delta}\right)^{O(1)}$-cut.


\subsection{Partitioning the Plane}

Suppose we are given a set $\T$ of
pairwise non-touching
triangles. Denote by $\E$ the
set of all boundary edges of the triangles in $\T$. For constructing
the cut, we associate a weight $\bar{w}(p_{i,j})$ with each vertex
$p_{i,j}$ of a triangle from $\T$. For each triangle $T_{i}$, we
define $\wb(p_{i,j}):=w(T_{i})/3$ for each point $p_{i,j}$, i.e.,
we equally distribute the weight of $T_{i}$ among its vertices. For
any area $C\subseteq I$ we define $\wb(C):=\sum_{p_{i,j}\in C}\wb(p_{i,j})$.

We construct the partition in three steps: 
\begin{enumerate}
\item We subdivide the plane into $O(1/\delta^{2})$ vertical \emph{stripes}
such that each stripe contains points with total weight smaller than
$w(\T)\cdot\delta^{2}$. The stripes are disjoint open sets. 
\item We subdivide each vertical stripe, along the lines from $\E$ crossing
the stripe from left to right, into $O(1/\delta^{4})$ \emph{cells}.
Each cell has either no lines in $\E$ crossing it from left to right,
or all triangles touching the cell have total weight at most $w(\T)\cdot\delta^{4}$.
The cells, similarly as stripes, are open sets. 
\item We transform the collection of cells into a subdivision of the plane
such that each face of the subdivision surrounds exactly one triangle or touches
triangles of total weight at most $O(\delta^2 \cdot w(\T))$, and each triangle in $\T$
is intersected only $O(1)$ times. 
\end{enumerate}
Later, we will apply the planar 
graph separator theorem from~\cite{Arora1998}
to the subdivision to obtain a cut with the desired properties.


\subsubsection{Partition into Cells }

We construct the vertical stripes as maximal stripes of the input
square which contain points of total weight smaller than $\delta^{2}\cdot w(\T)$.

We proceed iteratively as follows. First, we define $x_{0}:=0$ and
we set $x_{1}$ such that $\wb((x_{0},x_{1})\times[0,N])<\delta^{2}\cdot w(\T)$
and $\wb((x_{0},x_{1}]\times[0,N])\ge\delta^{2}\cdot w(\T)$. Such
$x_{1}$ is uniquely defined. We create a stripe $S_{1}:=(x_{0},x_{1})\times[0,N]$.
Iteratively, we define $x_{i}$ such that $\wb((x_{i-1},x_{i})\times[0,N])<\delta^{2}\cdot w(\T)$
and $\wb((x_{i-1},x_{i}]\times[0,N])\ge\delta^{2}\cdot w(\T)$ and
we set $S_{i}:=(x_{i-1},x_{i})\times[0,N]$. We stop when we have
defined a value $x_{i^{*}-1}$ such that $\wb((x_{i^{*}-1},N)\times[0,N])<\delta^{2}\cdot w(\T)$.
In that case we set $x_{i^{*}}=N$, we define $S_{i^{*}}:=(x_{i^{*}-1},x_{i^{*}})\times[0,N]$
and we terminate. We obtain the following property.

\begin{prop} 
There is a subdivision of the input square $I$ into
a set of vertical stripes $S_{1},...,S_{i^{*}}$ such that $i^{*}\le1/\delta^{2}$
and for each stripe $S_{i}$ we have $\wb(S_{i})<\delta^{2}\cdot w(\T)$.
\end{prop}

As a next step, we subdivide each stripe $S_{i}$, along the lines
of $\E$ crossing $S_{i}$ from left to right, into a collection
of \emph{cells}. We define the subdivision in such a way that
each cell has either no lines in $\E$ crossing it from left
to right, or all triangles touching the cell have total weight at
most $w(\T)\cdot\delta^{4}$.

Consider a stripe $S_{i}=(x_{i-1},x_{i})\times[0,N]$. Each cell for
the stripe $S_{i}$ will be a polygon with four edges, whose four consecutive
corners are of the form $(x_{i-1},y_{1}^{L}),(x_{i-1},y_{2}^{L}),(x_{i},y_{2}^{R}),(x_{i},y_{1}^{R})$
for some values $y_{1}^{L},y_{2}^{L},y_{1}^{R},y_{2}^{R}$ such that
$L[(x_{i-1},y_{1}^{L})(x_{i},y_{1}^{R})]$ and $L[(x_{i-1},y_{2}^{L})(x_{i},y_{2}^{R})]$
are subsegments of lines in $\E$, or of the boundary of $I$ (see Figure~\ref{fig:stripes-cells}).
We define the set $\E_{i}$ to be the set of all lines in $\E$ which cross $S_{i}$
from left to right. The subdivision of $S_i$ into cells is obtained by selecting a set
$\bar{\E}_{i}\subseteq\E_{i}$ whose elements then yield the top and bottom boundaries of the
cells. This selection is done in a straight-forward manner.

\begin{figure}
\begin{centering}
\includegraphics[scale=0.60]{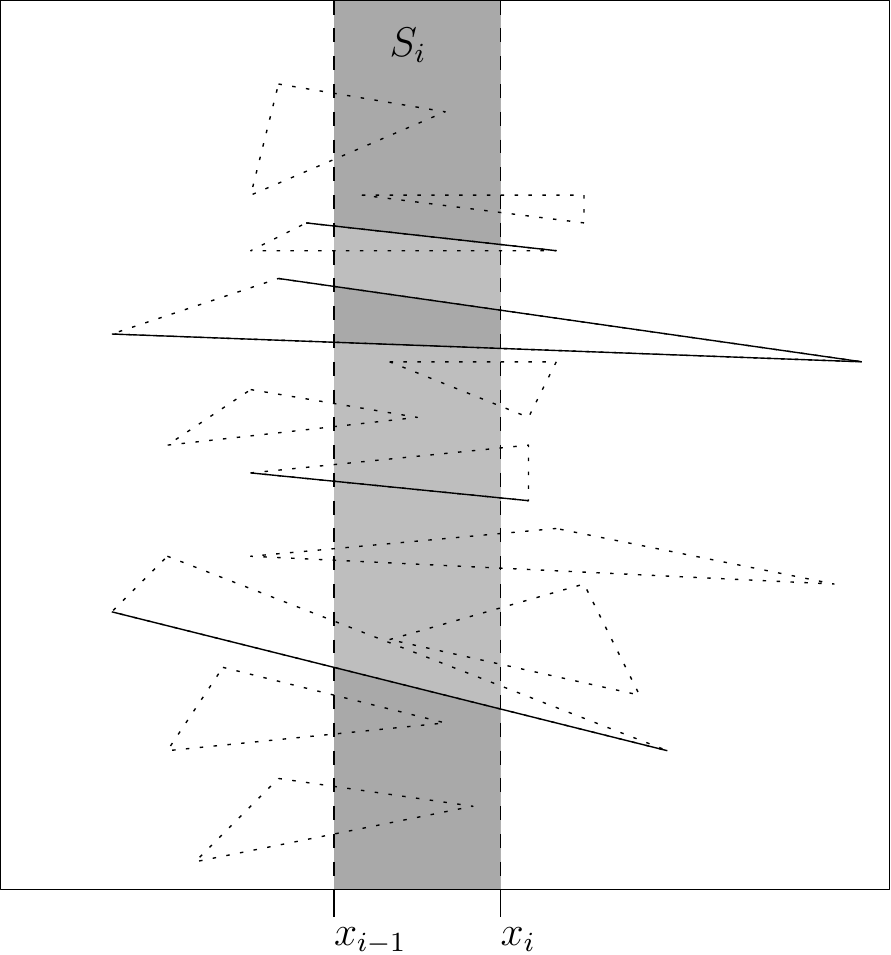}~~~~~~~~~~~~~\includegraphics[scale=1.0]{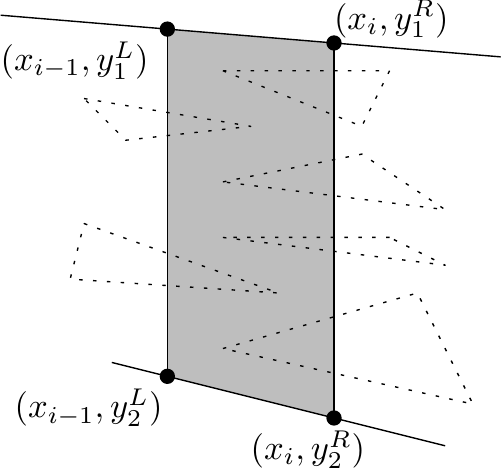} 
\par\end{centering}

\caption{\label{fig:stripes-cells}Left: Subdivision of a stripe $S_{i}$ into cells. The dotted lines denote
the triangles from $\T$ touching $S_{i}$, and the solid lines denote the edges
in $\bar{\E}_{i}$. The dark gray areas are dense cells, the light gray areas are light cells of $S_i$. Right: A dense cell. The solid vertical lines indicate
the lines in $\Le$ added for this cell.
}
\end{figure}

\begin{lem} \label{lem:partition-stripes} 
For any stripe $S_i$ there is a set $\bar{\E}_{i}\subseteq\E_{i}$
with $|\bar{\E}_{i}|\le O(\frac{1}{\delta^{4}})$ such that for any
connected component $C$ of $S_{i}\setminus \bar{\E}_{i}$
we have that $C\cap \E_{i}=\emptyset$ or the total weight of triangles touching $C$ is at most $\delta^{4}\cdot w(\T)$. 
\end{lem}

We define $\L_{0}:=\cup_{i}\bar{\E}_{i}$. These lines subdivide the
stripes into cells. An area $C\subseteq I$ is a \emph{cell} if there
is a stripe $S_{i}$ such that $C$ is a connected component of 
$S_i \setminus\bar{\E}_{i}$.
We call a cell $C$ \emph{dense} if $C\cap\E_{i}=\emptyset$
and \emph{light} otherwise, see Figure \ref{fig:stripes-cells}.
As the endpoints of all lines in $\L_{0}$ are vertices of the input polygons, they are basic DP-points.

\begin{prop} \label{prop:number-lines-L_0}
The number of lines in $\L_{0}$ is upper bounded by $O(1/\delta^{6})$. 
The endpoints of lines in $\L_{0}$ are basic DP-points.
\end{prop}


\subsubsection{Creating a Subdivision of the Plane}

Starting with the lines in $\L_{0}$ we create a partition of the
plane with some useful properties that will allow us later to find
a balanced cheap cut. To this end, we construct a set of lines $\Le$
that complete $\L_{0}$ to a subdivision of $I$. The endpoints of all lines in $\Le$ will be at
basic DP-points. Initially, we define $\Le$
to be the four boundary lines of $I$.

Recall that we have two types of cells: dense and light cells. First, we consider the dense cells, and
for each dense cell $C$ we perform the following operation.
Suppose that the consecutive corners of $C$ are the points $(x_{i-1}y_{1}^{L})$, $(x_{i-1},y_{2}^{L})$, $(x_{i},y_{2}^{R})$, $(x_{i},y_{1}^{R})$, see Figure~\ref{fig:stripes-cells}. We add to $\Le$ the two vertical lines $L[(x_{i-1},y_{1}^{L}),(x_{i-1},y_{2}^{L})]$
and $L[(x_{i},y_{1}^{R}),(x_{i},y_{2}^{R})]$. 

In the second step, we take each line $L\in\L_{0}$ and add a sequence
of lines to $\Le$ in order to connect both endpoints of $L$ with
another line in $\L_{0}\cup\Le$. While doing this we carefully ensure
that the number of lines in $\Le$ is upper bounded by $O(1/\delta^{8})$.

\begin{figure}
\begin{centering}
\includegraphics[scale=0.60]{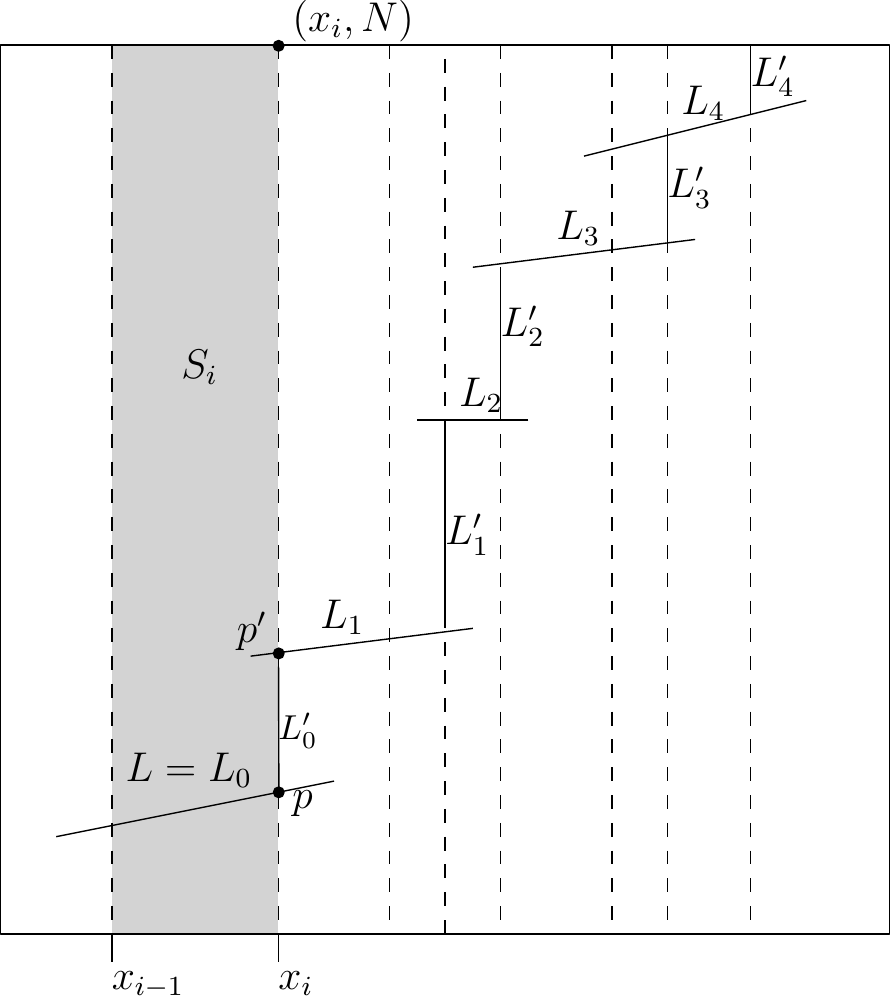}~~~~~~~~~~~~~\includegraphics[scale=1.2]{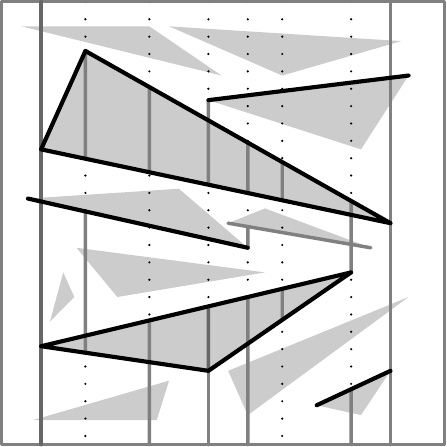}
\par\end{centering}

\caption{\label{fig:cells} Left: A sequence of lines
$L'_{0},L_{1},L'_{1},...$ that connect a line $L \in \L_{0}$ with the 
top boundary of~$I$. 
Right: The set of lines $\L_{0} \cup \Le$ subdividing the input area. 
The the lines from $\L_{0}$ are drawn in black, the lines from $\Le$ are gray. The gray lines in the interior of the two large triangles will be removed while transforming $\L_{0} \cup \Le$ into $\L$.}
\end{figure}

Let $L_{0}\in\L_{0}$. The idea is to define a set of lines which
connect $L_{0}$ to the right with either another line in $\L_{0}$
or with the top or right boundary of~$I$. Let $x_{i}$ be
the rightmost value in $\{x_{1},...,x_{i^{*}}\}$ such that $L_{0}$
touches the vertical line $\{x_{i}\}\times[0,N]$, and let $p$ be
the point where $L_{0}$ touches
this line. If $p$ is contained
in the top or right boundary of~$I$ or in a line $\bar{L}\in\L_{0}$
that crosses $S_{i+1}$, we stop. Now suppose that the line
$L':=L[p,(x_{i},N)]$ does not touch an edge in $\E$ that crosses
$S_{i+1}$. In that case we define $L'_{0}:=L'$ and we stop. Otherwise,
suppose that $L'$ touches an edge $L''\in\E$ crossing $S_{i+1}$
at a point $p'$ and $p'$ is the bottom-most point with that property
(possibly $p'=p$). We define $L'_{0}:=L[p,p']$.
If $L''\in\L_{0}$, or if $p'$ is at the top boundary of $I$, we stop. Otherwise, we
define $L_{1}:=L''$ and we continue iteratively with $L_{1}$. Eventually,
we have obtained a collection of lines $L_{0},L'_{0},L_{1},L'_{1},...,L_{\ell},L'_{\ell}$
where $L'_{\ell}$ touches the outside boundary of $I$ or a line in
$\L_{0}$. Note that all lines $L'_{i}$ are vertical. We add the
set $\{L'_{0},L_{1},L'_{1},...,L_{\ell},L'_{\ell}\}$ to $\Le$. See Figure~\ref{fig:cells}
for a sketch.

Similarly, we connect $L_{0}$ to the left with another line in $\L_{0}\cup\Le$
or with the bottom or left boundary of~$I$. 
We do this procedure with every line $L\in\L_{0}$. In case that $\Le$
contains two lines $L,L'$ with $|L\cap L'|>1$ (i.e., they are parallel
and overlap), we replace $L$ and $L'$ by $L\cup L'$.


\begin{prop}\label{prop:number-lines-L_ext} 
The number of lines in $\L_{0}\cup\Le$ is upper bounded by $O(1/\delta^{8})$.
The endpoints of lines in $\L_{0}\cup\Le$ are basic DP-points.
\end{prop}

The lines in $\L_{0}\cup\Le$ give a subdivision of the input square
$I$ (see Figure~\ref{fig:cells}). By construction, the only lines that touch triangles from $\T$ are vertical
lines in $\Le$, 
and if a line $L \in \Le$ touches a triangle $T \in \T$, then $L$ crosses $T$ 
(recall here that the triangles are open sets).
We want to ensure that each triangle is crossed only
$O(1)$ times. However, right now it can happen that a triangle is crossed
a superconstant number of times. We will
show that if a triangle $T$ is crossed by more than four lines
in $\Le$, then there must be a dense cell $C$ such that $C\subseteq T$.
In such a case we say that \emph{$T$ owns the cell $C$}.
\aad{We denote by $\T_{\own}$ the set of triangles from $\T$ which own some cell.} 

\begin{lem}\label{lem:cell-owning} Let $T\in\T$ be a triangle which
is crossed by more than four lines in $\Le$. Then 
\aad{$T \in \T_{\own}$.}
\end{lem}

\begin{proof} 
Suppose that there
is a triangle $T\in\T$ that is crossed by five lines in $\Le$. As
$T$ has three corners, there must be a vertical line $L\in\Le$ crossing
$T$ with an $x$-coordinate $x_{i}$ such that neither $S_{i}$ nor
$S_{i+1}$ contain a corner of $T$. By construction, $L$ cannot
have been added when connecting the lines in $\L_{0}$ with the boundary
of $I$. Hence, $L$ has been added due to a dense cell $C$ with $C\subseteq S_{i}$
or $C\subseteq S_{i+1}$. Assume the former is true. As $S_{i}$ does
not contain a corner of $T$, two edges of $T$ must cross $S_{i}$.
As $T\cap S_{i}\subseteq C$ and $C$ is a dense cell, this implies
that $T\cap S_{i}=C$. 
\aad{Thus, $T$ owns $C$ and therefore $T \in \T_{\own}$.}
\end{proof}

Since the total number of cells, similarly as the total number of lines in $\L_0$, is upper bounded by $O((1/\delta)^{6})$ (see Proposition \ref{prop:number-lines-L_0}),
this is also an upper bound on the number of triangles that are crossed
by more than four lines. We change the subdivision given by the lines
$\L_{0}\cup\Le$ by ``cutting out'' every triangle 
\aad{$T \in \T_{\own}$.}
For constructing the new set of lines $\L$, let $\L'$ be the set of all boundary edges of the triangles in $\T_{\own}$, together with the line segments given by $L\setminus\bigcup_{T\in\T_{\own}}T$
for each line $L\in\L_{0}\cup\Le$. We add to $\L$ the lines resulting from subdividing $\L'$ so that any two lines can touch only at their endpoints.
\aad{
Formally, we add to $\L$ each line $L[p,p']$ such that $L[p,p']$ is contained in a line in $\L'$, each of the points $p$ and $p'$ is either an endpoint of a line in $\L'$, or a point where two lines in $\L'$ touch, and no line in $\L'$ touches $L[p,p']$ at any point other than $p$ and $p'$.
}


\begin{prop}\label{prop:number-lines-L} 
The number of lines in $\L$ is upper bounded by $(\frac{1}{\delta})^{O(1)}$. 
\aad{The endpoints of lines in $\L$
are basic DP-points.}
\end{prop}

The lines in $\L$ induce a connected graph which gives us a subdivision of the plane.
We call each
connected component of $I\setminus\L$ a \emph{face}. Let $\F$ denote
the set of all faces. We prove now some crucial properties of $\F$
that will allow us later to find a balanced cheap cut 
as a part of the subdivision given by~$\L$. 
One important property
of the subdivision is that for each face $F\in\F$ we can bound the
weight of the triangles touching it, unless $F$ coincides with some
triangle $T \in \T$.

\begin{lem}\label{lem:weight-of-faces} Consider a face $F\in\F$.
Either $F=T$ for some triangle $T$ or the total weight of triangles
touching $F$ is upper bounded by $3\delta^{2}\cdot w(\T)$.
\end{lem}

\begin{proof}[Proof sketch.]
If $F$ does not coincide with any triangle $T$, then there are two
cases. The first case is that $F$ touches a dense cell $C$. In that
case, due to the vertical lines added for $C$, the face $F$ is in
fact contained in~$C$. As no triangles cross $C$, the weight of triangles
touching $F$ is bounded by $3\cdot\wb(C)\le3\delta^{2}\cdot w(\T)$.
If $F$ does not touch any dense cell, then one can show that it touches
at most one (light) cell per stripe $S_{i}$. As there are at most $1/\delta^{2}$
stripes and the total weight of triangles touching any light cell
is at most $\delta^{4}\cdot w(\T)$, the weight of triangles touching
$F$ is at most $\delta^{2}\cdot w(\T)$.
\end{proof}

The second 
important property of our subdivision is that each triangle is 
touched by at most a constant number of lines in $\L$. 

\begin{lem}\label{lem:bounded-crossing-number} For each triangle
$T\in\T$ there are at most four lines $L\in\L$ touching $T$. \end{lem}

\begin{proof} 
From the construction of $\L$ we know that if a triangle is touched by a line $L\in\L$, 
then $T$ is crossed by $L$. Moreover, the only lines from $\L$ which can cross triangles 
from $\T$ are the lines which were originally in $\Le$. 
If a triangle has been crossed by more than four lines of $\Le$ then it owned
a cell and has been cut out when constructing $\L$ (see Lemma~\ref{lem:cell-owning}).
Hence, for each triangle $T\in\T$ there are at most four lines $L\in\L$
touching $T$.
%
\end{proof}


\subsection{Obtaining the Cut}

We use the subdivision given by the lines in $\L$ in order to create
a balanced cheap cut. We construct a planar graph as follows. For
each point $p$ that is an endpoint of a line in $\L$ we create
a vertex $v_{p}$. By Proposition~\ref{prop:number-lines-L} 
such a point~$p$ must be a basic DP-point.
For each line $L[p,p'] \in \L$ we create an edge $e_{p,p'}:=\{v_{p},v_{p'}\}$. We define the
cost $c_{p,p'}$ of each edge $e_{p,p'}$ to be the total weight of
triangles in $\T$ touching $L[p,p']$.

For defining the weights of the faces, if a triangle $T\in\T$ has
non-empty intersection with $m$ faces of the graph, then each of
these faces obtains a $1/m$-fraction of $w(T)$. Denote by $G=(V,E)$
the resulting graph.
The proposition below follows from Proposition 
\ref{prop:number-lines-L} and Lemmas~\ref{lem:weight-of-faces}
and \ref{lem:bounded-crossing-number}, using that $\delta < 1/3$.

\begin{prop}\label{prop:graph} 
The total cost of the edges $E$ is bounded by $O(w(\T))$
and the total weight of the faces of $G$ is $w(\T)$. The total number
of edges is bounded by $(\frac{1}{\delta})^{O(1)}$. The cost of each
face is bounded by $\frac{1}{3}\cdot w(\T)$. 
The boundary of each face touches triangles of weight at most $3 \delta^2 \cdot w(\T)$.
\end{prop}

For extracting a balanced cheap cut from $G$ we use the following
theorem due to Arora et al.~\cite{Arora1998}. We note that
this theorem has been used in a similar way in~\cite{AW2013}.
A \emph{V-cycle} $C$
is a Jordan curve in the embedding of a given planar graph $G$ which
might go along the edges of $G$ and also might cross faces of $G$.
The parts of $C$ crossing an entire face of $G$ are called \emph{face
edges.}

\begin{thm}[\cite{Arora1998}] \label{thm:cycle-separator} Let
$G$ denote a planar, embedded graph with weights on the vertices
and faces and with costs on the edges. Let $W$ denote the total weight,
and $M$ the total cost of the graph. Then, for any parameter $\bar{k}$,
we can find in polynomial time a separating V-cycle $C$ such that 
the interior and exterior of $C$ each has weight at most $2W/3$, 
$C$ uses at most $\bar{k}$ face edges, and 
$C$ uses ordinary edges of total cost $O(M/\bar{k})$. 
\end{thm}

We apply Theorem~\ref{thm:cycle-separator} with $\bar{k}=O(1/\delta)$
to our graph $G$,
which yields the proof of the following lemma. In combination with 
Lemmas~\ref{lem:good-cut-suffices} and \ref{lem:reduction-to-triangles} it then
yields our main theorem.

\begin{lem} \label{lem:cheap-cut}Let $\delta>0$ and let $\T$ be
a set of 
pairwise non-touching
triangles in the plane such that $w(T)<w(\T)/3$ for each
$T\in\T$. Then there exists a balanced $O(\delta)$-cheap $(\frac{1}{\delta})^{O(1)}$-cut
with corners at basic DP-points. \end{lem}


\begin{thm}\label{thm:main}
Let $K\in\mathbb{N}$. There exists a $(1+\eps)$-approximation
algorithm with 
\aad{a} 
running time of at most 
$(nK)^{\left(\frac{1}{\eps}K\cdot\log n \right)^{O(1)}}$
for the 
\aad{Maximum Weight Independent Set of Polygons} 
problem, assuming that each polygon has at most $K$ vertices. 
Hence, for the case that  $K\le (\log n)^{O(1)}$ there is a QPTAS for the 
\aad{Maximum Weight Independent Set of Polygons} 
problem.
\end{thm}





\newpage
  
\bibliographystyle{abbrv}
\bibliography{MWISP-citations}


\newpage

\appendix

\section{Position of Input Points}\label{apx:general-position}

We show that w.l.o.g.~we can assume that no three points
of the input are lying on a straight line (and, as a consequence, no two points of the input are coincident).

\begin{lem} Let $K,N\in\mathbb{N}$ and let $\P$ be a set of $n$
polygons on the plane with at most $K$ vertices each such that each
vertex has integer coordinates in $\{0,...,N\}$ and no two polygons
are identical. In polynomial time we can compute a set of polygons
$\P'$ with at most $K$ vertices each with integer coordinates in
$\{0,...,\poly(K,N,n)\}$ and a bijection $f:\P\rightarrow\P'$ such
that for any $P,P'\in\P$ the polygons $P$ and $P'$ touch
if and only if $f(P)$ and $f(P')$ touch. \end{lem}

\begin{proof} First we show that for any point $p$ with integral
coordinates in $\{0,...,N\}$ and any line $L$ going through two points with
integral coordinates in $\{0,...,N\}$ it holds that either $p$ lies on $L$ or the
distance between $p$ and $L$ is at least 
$c := \frac{1}{2\cdot N}$.

Let $L$ be a line going through two points $(x_{1},y_{1}),(x_{2},y_{2})$
with $x_{1},x_{2},y_{1},y_{2}\in\{0,...,N\}$ and let $p$ be a point
with integer coordinates such that $p$ does not lie on $L$. If $L$
is either vertical or horizontal, then the distance between $L$ and
$p$ is at least one. Otherwise, consider an arbitrary point $p'\in L$
such that at least one coordinate of $p'$ is integral. Then the other
coordinate of $p'$ is either of the form $x_{1}\pm h\cdot\left|\frac{x_{2}-x_{1}}{y_{2}-y_{1}}\right|$
or $y_{1}\pm h'\cdot\left|\frac{y_{2}-y_{1}}{x_{2}-x_{1}}\right|$
for $h,h'\in\mathbb{N}$. As $1\le|y_{2}-y_{1}|\le N$ and $1\le|x_{2}-x_{1}|\le N$,
this implies that the distance between $p$ and $L$ in the Manhattan
metric is at least $1/N$. In particular, for the point $\tilde{p}\in L$
that minimizes the distance between $p$ and $L$ in the Euclidean
metric, the distance between $\tilde{p}$ and $p$ in the Manhattan
metric is at least $1/N$ and thus their distance in the Euclidean
metric is at least 
$\frac{1}{2\cdot N}$.

\begin{figure}[t]
\begin{centering}
\includegraphics{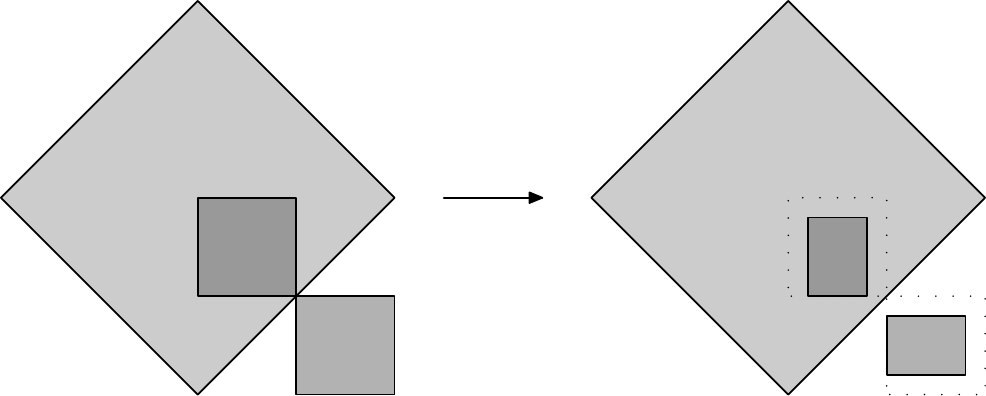} 
\par\end{centering}

\caption{\label{fig:no-3-on-a-line-1}Shifting vertices of the polygons
so that no three vertices are lying on one line. Some vertices are
moved by a small distance towards the interior of their polygons.
As a consequence of this operation, no two vertices of the polygons
are coincident. The combinatorial structure of the input remains the
same.}
\end{figure}

As we will show, if we move \emph{every} vertex $p$ of \emph{every
}polygon $P\in\P$ by a distance of at most $c':=c/5$ towards the
interior of the respective polygon $P$, we do not change the combinatorial
structure of $\P$ (see Figure~\ref{fig:no-3-on-a-line-1}). 
We can assume that all input vertices have even integer coordinates (otherwise, we could scale up all coordinates by a factor of two).
Observe that this implies that the mid-point of every edge has integer coordinates. 

As the polygons in $\P$ have at most $K\cdot n$ vertices in total,
there are at most $(K\cdot n)^{2}$ possible lines going through two
of these vertices. For each of those lines $L$ we want to ensure
that no third vertex lies on $L$. We achieve this as follows. We
consider the vertices of the polygons one by one. 
Let $p$ be a vertex
of a polygon $P$, and let $e,e'$ be the edges of $P$ adjacent to $p$.
We will move $p$ towards the interior of $P$.
Let $\vec{v}$ denote the vector representing $e$, i.e., $p+\vec{v}$
yields the coordinates of the other endpoint of $e$. Similarly, denote
by $\vec{v}'$ the vector representing $e'$. We move $p$ to the
point 
$p+(\ell\cdot\vec{v}+\ell'\cdot\vec{v}')\frac{c'}{10}\cdot\frac{1}{4N(Kn)^{2}}$
for some integers $\ell,\ell'\in\{0,...,(Kn)^{2}-1\}$ 
if the angle of the polygon at $p$ is smaller than $\pi$, and to the point 
$p-(\ell\cdot\vec{v}+\ell'\cdot\vec{v}')\frac{c'}{10}\cdot\frac{1}{4N(Kn)^{2}}$
if the angle at $p$ is greater than $\pi$ (see Figure~\ref{fig:no-3-on-a-line-2}).
In total, $p$ moves by a distance of at most $c'/5$ since 
$\Vert\vec{v}\Vert\le4N$ and $\Vert\vec{v}'\Vert\le4N$. We claim that we can find values
for $\ell$ and $\ell'$ such that after moving $p$ the resulting
point does not lie on a line $L$ that goes through two other vertices
$p_{1},p_{2}$ of polygons in $\P$.

\begin{figure}[t]
\begin{centering}
\includegraphics{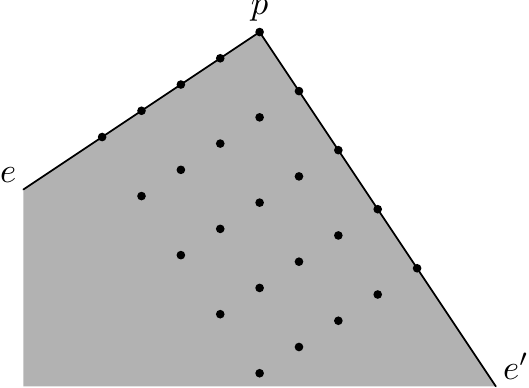}~~~~~~~~~~~~~\includegraphics[scale=.9]{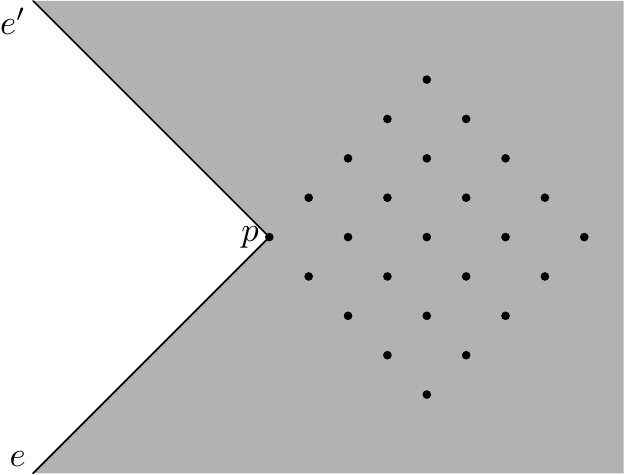}
\par\end{centering}

\caption{\label{fig:no-3-on-a-line-2}Possible positions where a vertex
$p$ of a polygon can be moved, when the angle at $p$ is smaller than $\pi$ (left) or greater than $\pi$ (right).}
\end{figure}

There are $(Kn)^{4}$ possible positions $p'$ where a point $p$ can be moved
(note that no
two different pairs $\{\ell_{1},\ell'_{1}\}$, $\{\ell_{2},\ell'_{2}\}$
yield the same point since the vectors corresponding to the edges
$e$ and $e'$ are linearly independent). Each pair $p_{1},p_{2}$
of two vertices of polygons in $\P$ with $p_{1}\ne p\ne p_{2}$ might
forbid us to use up to $(Kn)^{2}$ of these candidates and there are
at most $(Kn-1)\cdot(Kn-2)<(Kn)^{2}$ such pairs. Thus, by the pidgeonhole
principle, we can find a pair $\ell,\ell'$ such that $p$ does not
lie on any line going through two vertices $p_{1},p_{2}$ with $p_{1}\ne p\ne p_{2}$.
Performing this operation for every vertex $p$ of every polygon $P$ defines the map
$f$ (see the statement of the lemma).

We claim that after this modification two polygons $f(P),f(P')$ touch
if and only if $P$ and $P'$ touch. We moved all vertices of the
polygons towards the interior of the respective polygon. Thus, $P\cap P'=\emptyset$
implies that $f(P)\cap f(P')=\emptyset$.

Now suppose that $P\cap P'\ne\emptyset$. We want to show that
then $f(P)\cap f(P')\ne\emptyset$. We distinguish two cases. The
first case is that there is an edge $e$ of $P$ and an edge $e'$
of $P'$ such that $e$ and $e'$ cross each other. Note that then
the point $p:=e\cap e'$ is not an endpoint of either of the two.
By the above reasoning, the distance between $p$ and any of the four
endpoints of $e$ and $e'$ is at least $c$ (as the distance between
any vertex and any line going through two vertices of a polygon is
either zero or at least $c$). Since we moved each point by a distance
of at most $c'=c/5$, after moving the points $e$ and $e'$ still
cross each other and thus $f(P)\cap f(P')\ne\emptyset$. Assume now
that there are no such two edges. Then every vertex of the polygon
$P'':=P\cap P'$ coincides with a vertex of $P$ or a vertex of $P'$, 
see Figure \ref{fig:intersecting-polygons}.
\begin{figure}[t]
\begin{centering}
\includegraphics[scale=1]{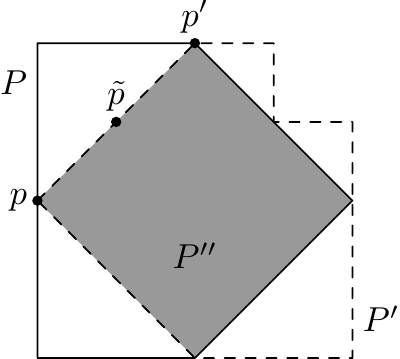}
\par\end{centering}
\caption{\label{fig:intersecting-polygons}If two polygons $P$ (solid boundary) and $P'$ (dashed boundary) touch, and there are no two edges $e$ and $e'$ of $P$ and $P'$, respectively, which cross each other, there must be a point $\tilde{p}$ on the boundary of one of the polygons 
\aad{(in this case $P'$), which lies in the interior of the other polygon
($P$) and which has integer coordinates.
Hence, it has a distance of
at least $c=\frac{1}{2N}$ to the boundary of the 
second polygon (P).}}
\end{figure}

As $P\ne P'$, there must be two points $p,p'$ that are vertices of $P$ or $P'$ such that
$L[p,p'] \in \partial P''$, 
the line segment $L[p,p']$ does not contain any vertex from $P$ or $P'$ apart from $p$ and $p'$, and $L[p,p']$ 
is not contained in any edge of $P$ or it is not contained in any edge of $P'$.
Suppose
w.l.o.g.~that $L[p,p']$ is not contained in any edge of $P$. 
Then the midpoint $\tilde{p}$
of $L[p,p']$  has integral coordinates and is contained in $P$.
This implies that $\tilde{p}$ has a distance of at least $c$ to
$\partial P$ and note that $\tilde{p}\in\partial P'$. Therefore,
after moving the polygons, there is a point on $\partial f(P')$ which
is contained in $f(P)$. Thus, $f(P)\cap f(P')\neq\emptyset$.

Initially, all vertices have integer coordinates. Then, each vertex
$p$ is moved to some point 
$p\aad{\pm}(\ell\cdot\vec{v} + \ell'\cdot\vec{v}')\frac{c'}{10}\cdot\frac{1}{4N(Kn)^{2}}$
for integers $\ell,\ell'\in\{0,...,(Kn)^{2}-1\}$. The resulting coordinates
are of the form $\tilde{\ell}\cdot\frac{1}{h(K,N,n)}$ for an integer
$\tilde{\ell}\in\{0,...,N\cdot h(K,N,n)\}$ for a polynomial $h$
that is independent of the problem instance. By scaling all coordinates up linearly
by a factor of $h(K,N,n)$, we obtain a set of polygons $\P'$ with integer
coordinates in $\{0,...,N\cdot h(K,N,n)\}$ that fulfill the claim
of the lemma.

The above operation can be performed efficiently. When moving each
vertex it suffices to identify the correct values for $\ell$ and $\ell'$,
and there is only a polynomial number of possibilities. 
\end{proof}

\section{Proof of Lemma~\ref{lem:good-cut-suffices}}

In order to prove Lemma~\ref{lem:good-cut-suffices} we show that
there is a family of polygons with corners at DP-points which
yield a recursion of GEO-DP with the claimed approximation
factor. In fact, the following two lemmas were almost identically
given in~\cite{AW2013}.

Note that, in the following, the polygons $\Gamma$ and $Q$ do not have to be simple, i.e., they can have holes.

\begin{lem} \label{lem:polygons-complexity} Let $\Gamma_{1},\ldots,\Gamma_{m}$
be a collection of polygons, where each $\Gamma_{i}$ has at most $\bar{\ell}$
edges. Then the intersection $\bigcap_{i=1,\ldots,m}\Gamma_{i}$ consists
of at most $(m\bar{\ell})^{2}$ connected components, and each connected
component is a polygon with at most $(m\bar{\ell})^{2}$ edges.
\end{lem}

\begin{proof} As the polygons $\Gamma_{i}$ have at most $m\bar{\ell}$
edges in total, and any two edges can cross at most once, we have
at most $(m\bar{\ell})^{2}$ pairs of crossing edges. Each connected
component of the intersection $\bigcap_{i=1,\ldots,m}\Gamma_{i}$ has at
least three corners, and each corner corresponds to a different pair
of crossing edges. Hence, there are at most $(m\bar{\ell})^{2}$ connected
components. The number of edges of one connected component $\Gamma$ equals
the number of vertices of $\Gamma$, which is also upper bounded by the
number of crossing pairs $(m\bar{\ell})^{2}$. \end{proof}

\begin{lem} \label{lem:families-Pj}Let $\alpha,\ell$ be as defined
in Lemma \ref{lem:good-cut-suffices} and let $\bar{\P}$ be a set
of at most $n$ pairwise non-touching polygons. Let $j^{*}=\left\lceil \log_{3/2}{n^{2}}/{\eps}\right\rceil $.
Then for each $j\in\left\{ 0,1,\ldots,j^{*}\right\} $ there is a
family of polygons $\Q_{j}$, such that:

\begin{enumerate}[a)] 
\item each polygon $Q\in\Q_{j}$ has at most $(j+1)^{2}\cdot(\ell+4)^{2}$
edges and all its vertices are DP-points,
\item $\Q_{0}=\{[0,N]\times[0,N]\}$, 
\item the polygons in $\Q_{j}$ are disjoint, and each polygon $Q\in\Q_{j-1}$
is a disjoint union of at most $(j+1)^{2}(\ell+4)^{2}$ polygons from
$\Q_{j-1}$, 
\item each polygon $Q\in\Q_{j^{*}}$ contains at most one polygon from $\bar{\P}$, 
\item for each set $\Q_{j}$ we get $\sum_{Q\in\Q_{j}}w(Q)\ge\left(1-\alpha\right)^{j}\cdot w(\bar{\P})$, where $w(Q)= \sum_{P \in \bar{\P}: P \subseteq Q} w(P)$. 
\end{enumerate}
\end{lem} 

\begin{proof} We set $\Q_{0}=\{[0,N]\times[0,N]\}$, i.e., $\Q_{0}$
consists of one square which contains all the polygons from $\bar{\P}$.
We then construct the sets $\Q_{1},\ldots,\Q_{j^{*}}$ one by one,
as follows. To construct $\Q_{j}$, we consider each polygon $Q\in\Q_{j-1}$,
and we add to $\Q_{j}$ the following set of polygons, which together
give a disjoint union of $Q$. If $Q$ contains at most one polygon
$P\in\bar{\P}$, we add $Q$ to the set $\Q_{j}$. Otherwise, if there
is a polygon $P_{0}\in\bar{\P}$, $P_{0}\subseteq Q$ with $w(P_{0})\ge\ 1/3\cdot w(Q)$,
we add to $\Q_{j}$ the following polygons: $Q\cap P_{0}=P_{0}$,
and the connected components of $Q\cap\overline{P_{0}}$, i.e. the
connected components of $Q\setminus P_{0}$. Finally, consider the
case that no polygon $P_{0}\in\bar{\P}$ with $P_{0}\subseteq Q$
has weight $w(P_{0})\ge\ 1/3\cdot w(Q)$. Then there exists a balanced $\alpha$-cheap
$\ell$-cut for the set of polygons from $\bar{\P}$ which are contained
in $Q$. Let $\Gamma$ be the polygon defining this cut, and let $\bar{\Gamma}$
be its complement intersected with the input square. Due to the assumptions
in Lemma~\ref{lem:good-cut-suffices}, all corners
of $\Gamma$ are basic DP-points. We add to $\Q_{j}$ all connected
components of $Q\cap\Gamma$ and $Q\cap\bar{\Gamma}$. Notice that
$\Gamma$ has at most $\ell$ edges, and so $\bar{\Gamma}$ has at
most $\ell+4$ edges.

We now have to check that all required properties are satisfied.

$a)$ The only polygon in $\Q_{0}$ has $4$ edges. Each polygon $Q\in\Q_{j}$
for $j\ge1$, is a connected component of an intersection of at most
$j+1$ polygons with at most $\ell+4$ edges each whose corners are basic DP-points. Hence, the vertices of
$Q$ are all DP-points and, due to Lemma
\ref{lem:polygons-complexity}, $Q$ has at most $(j+1)^{2}(\ell+4)^{2}$
edges.

$b)$ Defined at the beginning of the proof.

$c)$ From the construction of the sets $\Q_{j}$ it can be easily
observed that the polygons in $\Q_{j}$ are disjoint, and each polygon
$Q\in\Q_{j-1}$ is a union of polygons from $\Q_{j}$. We now have
to upper bound the number of polygons from $\Q_{j}$ which can be
contained in one polygon $Q\in\Q_{j-1}$. Polygon $Q$ is a connected
component of an intersection of at most $j$ polygons, each with at
most $\ell+4$ edges. By construction each polygon $Q'\in\Q_{j}$
is contained in some polygon $Q\in\Q_{j-1}$. Each polygon $Q'\in\Q_{j}$
contained in $Q$ is a connected component of an intersection of $Q$
with a polygon $\Gamma$ with at most $\ell$ edges, or with a polygon
$\bar{\Gamma}$ with at most $\ell+4$ edges. Therefore $Q'$ is a
connected component of an intersection of at most $j+1$ polygons,
each with at most $\ell+4$ edges.
From Lemma \ref{lem:polygons-complexity}
the number of such components is upper bounded by $(j+1)^{2}(\ell+4)^{2}$.

$d)$ For a polygon $Q$ let $w(Q):=\sum_{P\in\bar{\P}:P\subseteq Q}w(P)$
denote the weight of all polygons from $\bar{\P}$ contained in $Q$.
We will show by induction that if $Q\in\Q_{j}$ contains more than
one polygon from $\bar{\P}$, then $w(Q)\le\frac{2}{3}^{j}w(\bar{\P})\le\frac{2}{3}^{j}n^{2}/\eps$.
For $Q\in\Q_{j^{*}}$ that value is at most $1$ since $\frac{2}{3}^{j^{*}}n^{2}/\eps\le1$.
As each polygon from $\bar{\P}$ has weight at least $1$, $Q$
cannot contain more than one polygon.

For the (unique) polygon $Q\in\Q_{0}$ we have $w(Q)=w(\bar{\P})\le n^{2}/\eps$
(see Lemma \ref{lem:bound-weights}).
We assume by induction that the property holds for $\Q_{j-1}$, and
we will show that it holds also for $\Q_{j}$. Let $Q\in\Q_{j}$ be
contained in a polygon $Q_{0}\in\Q_{j-1}$, where $w(Q_{0})\le\frac{2}{3}^{j-1}w(\bar{\P})\le\frac{2}{3}^{j-1}n^{2}/\eps$.
If $Q$ contains more than one polygon from $\bar{\P}$, then either
$Q\subseteq Q_{0}\setminus P'$ for a polygon $P'$ with $w(P')\ge\frac{1}{3}\cdot w(Q_{0})$,
or $Q$ is obtained from $Q_{0}$ by a balanced cut. In both cases
we get $w(Q)\le\frac{2}{3} w(Q_{0}) \le \frac{2}{3}^j w(\bar{\P})$.

$e)$ The property holds for $\Q_{0}$, as $\sum_{Q\in\Q_{0}}w(Q)=w(\bar{\P})$.
We give a proof by induction. Assume that the property holds for $\Q_{j-1}$,
i.e., $\sum_{Q\in\Q_{j-1}}w(Q)\ge\left(1-\alpha\right)^{j-1}\cdot w(\bar{\P})$.
The polygons which are intersected by $\Q_{j}$, but not by $\Q_{j-1}$,
must be intersected by the newly introduced cuts $\Gamma$, which
intersect polygons $Q\in\Q_{j-1}$. As each cut $\Gamma$ is a $\alpha$-cheap
$\ell$-cut for the set of polygons contained in the corresponding
polygon $Q$, we get $\sum_{Q'\in\Q_{j}:Q'\subseteq Q}w(Q')\ge(1-\alpha)w(Q)$
for each $Q\in\Q_{j-1}$, and so $\sum_{Q\in\Q_{j}}w(Q)\ge\left(1-\alpha\right)^{j}\cdot w(\bar{\P})$.
\end{proof} 

With this preparation we are able to prove Lemma~\ref{lem:good-cut-suffices}. 

\begin{proof}[Proof of Lemma~\ref{lem:good-cut-suffices}]
We run GEO-DP, parametrized by $k:=\left(\left\lceil \log_{3/2}({n^{2}}/{\eps})\right\rceil +1\right)^{2}\cdot(\ell+4)^{2}$,  on $\bar{\P}$.
Denote by $\Q_{j}$, with $j\in\left\{ 0,1,\ldots,j^{*}\right\} $
for $j^{*}=\left\lceil \log_{3/2}({n^{2}}/{\eps})\right\rceil $,
the families of polygons with DP-points as coordinates as given by
Lemma~\ref{lem:families-Pj}.

From Lemma~\ref{lem:families-Pj}a) any polygon $Q\in\Q_{j}$ has
at most $k$ edges, all its vertices are DP-points,
and so GEO-DP has a DP-cell for $Q$. If $Q\in\Q_{j^{*}}$,
from Lemma~\ref{lem:families-Pj}d) we know that $Q$ contains at
most one polygon from $\bar{\P}$, and so $w(sol(Q))=w(Q)$, where for each polygon
$Q$ we denote by $w(Q)$ the total weight of all polygons from $\bar{\P}$
which are contained in $Q$. From Lemma \ref{lem:families-Pj}c) each
polygon $Q\in\Q_{j}$ is a union of at most $k$ polygons $Q_{1},\ldots,Q_{m}\in\Q_{j+1}$.
Therefore GEO-DP tries the subdivision of $Q$ into these components
and we get that $w(sol(Q))\ge\sum_{i=1}^{m}w(sol(Q_{i}))$, which
for the input polygon $Q_{0}\in\Q_{0}$ (see Lemma~\ref{lem:families-Pj}b))
gives 
\[
w(sol(Q_{0}))\ge\sum_{Q\in\Q_{j^{*}}}w(sol(Q))=\sum_{Q\in\Q_{j^{*}}}w(Q)\ge\left(1-\alpha\right)^{j^{*}}\cdot w(\bar{\P}),
\]
where the last inequality comes from Lemma~\ref{lem:families-Pj}e). The approximation ratio of GEO-DP is then $\left(1-\alpha\right)^{-j^{*}} = \left(\frac{1}{1-\alpha}\right)^{\left\lceil \log_{3/2}({n^{2}}/{\eps})\right\rceil} = (1+\alpha)^{O(\log(n / \eps))} = 1+O(\eps)$.
\end{proof}

\section{Other Omitted Proofs and Figures}\label{apx:proofs-figures}

\begin{proof}[Proof of Lemma~\ref{lem:bound-weights}]
Suppose we are given a set of polygons $\P$ with arbitrary weights.
First, we scale the weights of all 
\aad{polygons}
such that $\max_{P\in\P}w(P)=n/\eps$.
Since then $OPT\ge n/\eps$, all polygons $P' \in \P$ with
$w(P')<1$ can contribute a total weight of at most $n\cdot1=\eps\cdot\frac{n}{\eps}\le\eps\cdot OPT$.
We remove them from the instance which reduces the value of the optimal
solution by at most $\eps\cdot OPT$. Thus, the value of the optimal solution
decreases at most by a factor of $(1-\eps)^{-1}=1+O(\eps)$.
\end{proof}

\begin{proof}[Proof of Lemma~\ref{lem:partition-stripes}]
As the lines in $\E_{i}$ do not cross one another,
we can order them from top to bottom. We construct the set of lines
$\bar{\E}_{i}$ one by one. If $\E_{i}=\emptyset$ or if the total
weight of triangles touching $S_{i}$ is at most $\delta^{4}\cdot w(\T)$,
we set $\bar{\E}_{i}=\emptyset$ and we are done.

Let $e_{1}\in\E_{i}$ be the bottom-most line from $\E_{i}$ such
that the total weight of triangles touching $S_{i}$ above $e_{1}$
is at most $\delta^{4}\cdot w(\T)$. If such a line does not exist,
we set $e_{1}\in\E_{i}$ to be the top-most line crossing $S_{i}$.
Iteratively, we define $e_{j}\in\E_{i}$ be the bottom-most line from
$\E_{i}$ which is below $e_{j-1}$ and such that the total weight
of triangles touching $S_{i}$ between $e_{j-1}$ and $e_{j}$ is
at most $\delta^{4}\cdot w(\T)$. Again, if such a line does not exist,
we set $e_{j}\in\E_{i}$ to be the top-most line crossing $S_{i}$
below $e_{j-1}$. We stop when we have added to $\bar{\E}_{i}$ a
line $e_{j^{*}}\in\E_{i}$ such that the total weight of triangles
touching $S_{i}$ below $e_{j^{*}}$ is at most $\delta^{4}\cdot w(\T)$,
or when $e_{j^{*}}$ is the bottom-most line in $\E_{i}$.

Let $C$ be a connected component of $S_{i}\setminus\bigcup\bar{\E}_{i}$
such that $C\cap\bigcup\E_{i}\neq\emptyset$. Then from the construction
above the total weight of triangles touching $C$ is at most $\delta^{4}\cdot w(\T)$.

Let $e_{j},e_{j+1},e_{j+2}$ be three consecutive lines from $\bar{\E}_{i}$.
From the construction above we know that the total weight of triangles
touching $S_{i}$ between $e_{j}$ and $e_{j+2}$ is greater than
$\delta^{4}\cdot w(\T)$. As for each triangle $T\in\T$ the intersection
$T\cap S_{i}$ is contained between two consecutive lines from $\bar{\E}_{i}$
(or possibly above the first or below the last line from $\bar{\E}_{i}$),
we get that $|\bar{\E}_{i}|\le\frac{2}{\delta^{4}}$. 
\end{proof}

\begin{proof}[Proof of Proposition~\ref{prop:number-lines-L_ext}]
For each stripe $S_i$ we have $|\bar{\E}_i|\le O(1/\delta^4)$. Thus, the number of cells in 
\aad{a} 
stripe $S_i$ is at most $O(1/\delta^4)$. 
Since we have $O(1/\delta^2)$ stripes the total number of cells is upper bound by $O(1/\delta^6)$. 
For each dense cell we added two lines to $\Le$, so $O(1/\delta^6)$ in total.  
When connecting each of the $O(1/\delta^6)$ lines in $\L_0$ (cf. Proposition~\ref{prop:number-lines-L_0}) to the boundary of $I$ we added at most $O(1/\delta^2)$ lines per each line in $\L_{0}$. Thus,
\aad{$|\L_{0} \cup \Le| = O(1/\delta^8)$.}

Each endpoint of a line in $\Le$ created in the first step, i.e., while adding vertical boundary lines for dense cells, 
is a corner of a cell. 
Each endpoint of a line in $\Le$ created in the second step, i.e., while connecting lines from $\L_0$ with other lines, is either a vertex of an input triangle (this applies to the lines $L_1,L_2,\ldots$), or an intersection of an edge of a triangle 
with a vertical line $\{x_{i}\}\times[0,N]$ where the latter contains some vertex of an input triangle (this applies to the lines $L_0',L_1',\ldots$). By the definition of the DP-points each such point is a basic DP-point.
\end{proof}

\begin{proof}[Proof of Lemma~\ref{lem:weight-of-faces}]
Suppose that $F$ touches some dense cell $C \subseteq S_{i}$.
Due to the definition of the lines in $\L$ then either $F=T$ for
some triangle $T$ (when $T$ was cut out since it owned some cell,
not necessarily $C$) or $F\subseteq C$. If $F=T$, we are done. Now assume that $F\subseteq C$.
From the construction of the lines $\L_{0}\cup\Le$ and the definition
of a dense cell, $F$ is contained in the stripe $S_{i}$, $F$ is not contained in a single triangle $T \in \T$, and $F$
is not crossed by any line in $\E_{i}$. The two latter properties give us that each triangle
touching $F$ has at least one endpoint inside $F \subseteq S_{i}$. As $\wb(S_i) < \delta^2 w(\T)$, the
total weight of triangles touching $F$ is upper bounded by $3\delta^{2}\cdot w(\T)$.

Now assume that $F$ does not touch any dense cells and there is no
triangle $T$ such that $F=T$. We will show that for every stripe
$S_{i}$ the intersection $F\cap S_{i}$, if non-empty, is contained in one cell $C_i$. 
Let $C_{i}^{1}$ and $C_{i}^{2}$ be two cells of $S_i$. Let $L \in \L_{0}$ be a line crossing $S_i$ and separating $C_{i}^{1}$ from $C_{i}^{2}$ (see Figure~\ref{fig:C1-C2}). From the construction of the lines in  $\L_{0}\cup\Le$, there is a path of lines in $\L_{0}\cup\Le$ (and therefore also a path of lines in $\L$) containing $L$ and connecting the left or bottom boundary of $I$ with the top or right boundary of $I$. The cells $C_{i}^{1}$ and $C_{i}^{2}$ are separated by this path. Therefore any face $F$ touching $C_i^1$ does not touch $C_i^2$.

\begin{figure}[t]
\begin{centering}
\includegraphics[scale=0.8]{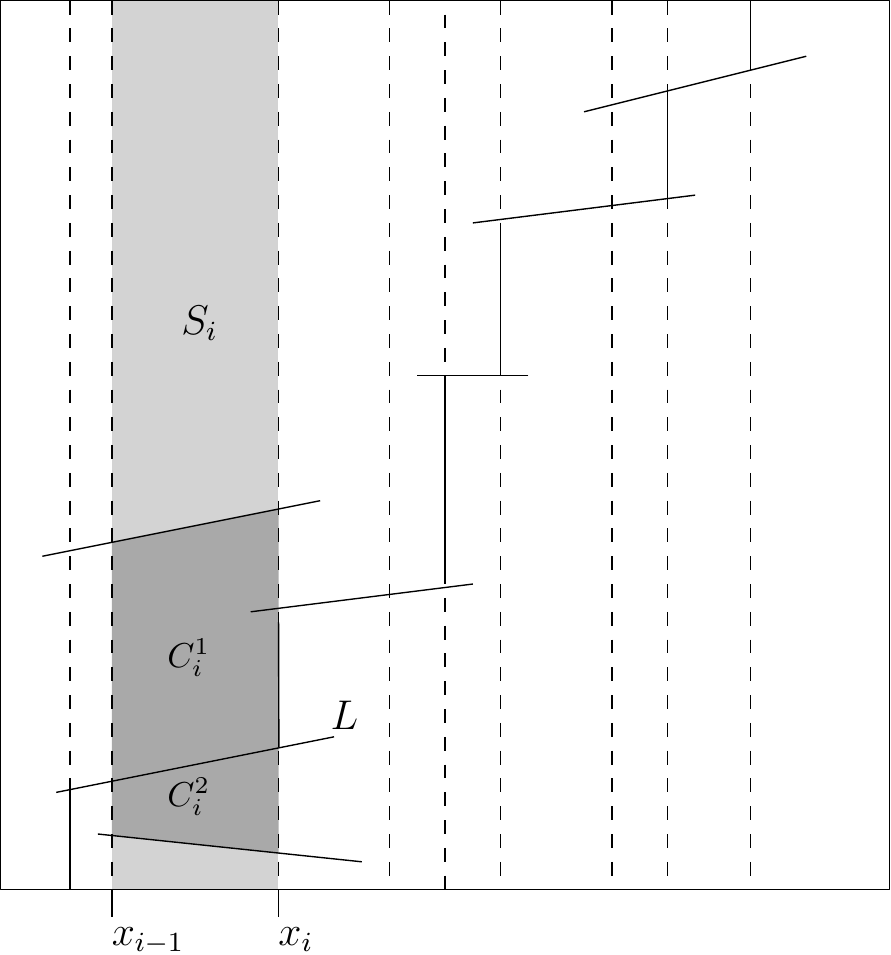}
\par\end{centering}
\caption{\label{fig:C1-C2}Two cells $C_i^1,C_i^2$ of the same stripe \aad{$S_i$} cannot lie in the same face of the partition.}
\end{figure}



As the
number of stripes is not greater than $1/\delta^{2}$ and from Lemma \ref{lem:partition-stripes}
for each light cell $C_{i}$ the total weight of triangles touching
$C_{i}$ is upper bounded by $\delta^{4}\cdot w(\T)$, the total weight
of triangles touching $F$ is upper bounded by $\delta^{2}\cdot w(\T)$.
\end{proof}

\begin{proof}[Proof of Proposition~\ref{prop:graph}] 
Due to Lemma~\ref{lem:bounded-crossing-number} each triangle $T \in \T$ is touched by at most four lines in $\L$. 

Thus, each triangle $T\in \T$ is touched by at most four edges of $G$ and hence the total cost of the edges $E$ is bounded by $O(w(\T))$.
The other claims follow directly from Lemmas~\ref{lem:weight-of-faces}
and \ref{lem:bounded-crossing-number} and Proposition \ref{prop:number-lines-L}, using that $\delta < 1/3$.
\end{proof}

\begin{proof}[Proof of Lemma~\ref{lem:cheap-cut}]
We apply Theorem~\ref{thm:cycle-separator} to the
graph $G=(V,E)$ obtained from the subdivision $\L$, with parameter $\bar{k}:=\frac{1}{\delta\cdot w(\T)}\sum_{e_{p,p'} \in E}c_{p,p'}$.
From Proposition \ref{prop:graph} we know that $\bar{k}=O(1/\delta)$, and that the obtained
V-cycle $C$ uses at most $O(1/\delta)$ face edges and
ordinary edges with total cost at most $O(\delta\cdot w(\T))$. 
We transform $C$ into a V-cycle $C'$ using only ordinary edges of $G$, by replacing each face edge connecting two vertices $v_p$ and $v_{p'}$ with a path $R(v_p,v_{p'})$ lying on the boundary of the crossed face. 
\aad{As due to Proposition~\ref{prop:graph} }each face has weight at most $w(\T)/3$, while performing the above operation we ensure that the modified cut $C'$ remains balanced (i.e., each side of the cut contains triangles of total weight at most $2/3 \cdot w(\T)$). We do that by choosing the path $R(v_p,v_{p'})$ in one of the two possible ways, depending on whether we want to connect the face with the interior, or with the exterior of $C'$. Notice that after this operation $C'$ might not be a simple cycle. However, we can modify $C'$ so that each edge appears only $O(1)$ times in $C'$. 
As from Proposition~\ref{prop:graph} the number of edges in $G$ is bounded by $(\frac{1}{\delta})^{O(1)}$, $C'$ uses at most $(\frac{1}{\delta})^{O(1)}$ edges.
All vertices of $C'$ are vertices of $G$, and so they are basic DP-points.

It remains to show is that $C'$ is $O(\delta)$-cheap. The ordinary edges of $C$ have cost $O(\delta\cdot w(\T))$. The remaining edges of $C'$ lie on the boundary of $O(1/\delta)$ faces of $G$, and from Proposition~\ref{prop:graph} they intersect triangles of total weight $O(\delta\cdot w(\T))$. Thus, the cycle $C'$ is $O(\delta)$-cheap.
\end{proof}

\begin{proof}[Proof of Theorem~\ref{thm:main}]
From Lemma \ref{lem:cheap-cut}, when setting $\delta:=\Theta(\frac{\epsilon}{K\cdot\log(n/\epsilon)})$, we obtain that for any set $\T$ of 
pairwise non-touching triangles in the plane such that $w(T)<w(\T)/3$ for each
$T\in\T$ there exists a balanced $\frac{\epsilon}{K\cdot\log(n/\epsilon)}$-cheap $(\frac{K\cdot\log(n/\epsilon)}{\epsilon})^{O(1)}$-cut
with corners at basic DP-points.
Applying Lemma \ref{lem:reduction-to-triangles} gives that for any set of pairwise non-touching
polygons $\P$ with at most $K$ edges each such that $\max_{P\in\P}w(P)<w(\P)/3$
there is a balanced $\frac{\epsilon}{\log(n/\epsilon)}$-cheap
$K\cdot (\frac{K\cdot\log(n/\epsilon)}{\epsilon})^{O(1)}$-cut
with vertices at basic DP-points.
From Lemma \ref{lem:good-cut-suffices}, GEO-DP parametrized by $k=(\frac{K\cdot\log(n/\epsilon)}{\epsilon})^{O(1)}$
has an approximation ratio of $1+O(\eps)$.
According to Proposition \ref{prop:running-time-GEO-DP}, the running time of GEO-DP is upper bounded by $(nK)^{O(k^{2})} = (nK)^{\left(\frac{1}{\eps}K\cdot\log(n/\eps)\right)^{O(1)}} \le (nK)^{\left(\frac{1}{\eps}K\cdot\log(n)\right)^{O(1)}}$.
If $K\le (\log n)^{O(1)}$ this yields a running time of $n^{(\frac{1}{\epsilon}\cdot \log n)^{O(1)}}$ which is quasipolynomial.
\end{proof}

\begin{figure}[h]
\begin{centering}
\includegraphics[scale=0.66]{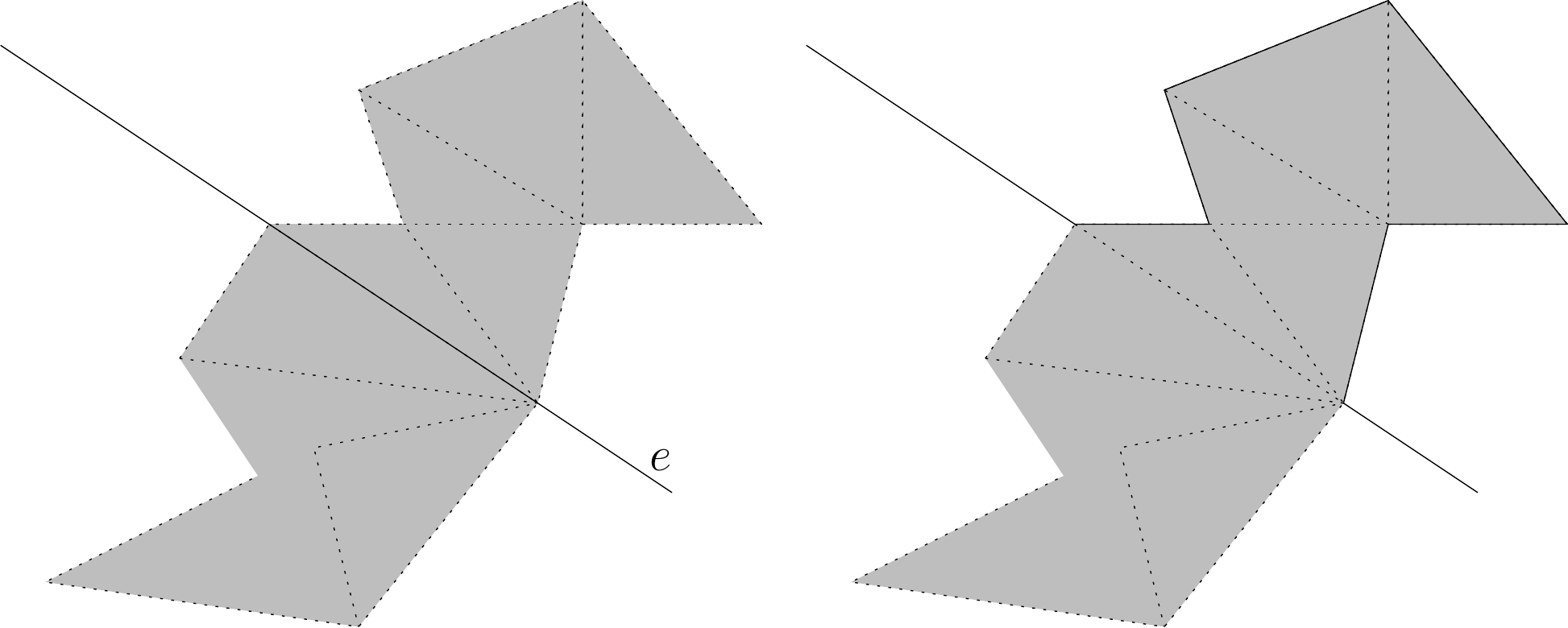} 
\par\end{centering}

\caption{\label{fig:circumvention-of-polygon}Circumvention of a polygon $P(e)$
\aad{(gray area). The dotted lines indicate the triangulation of the polygon.}
The left figure shows an edge $e$ in the cut $\Gamma$ such that
$e$ crosses $P(e)$, but it does not cross any triangle from $\T(P(e))$.
The right figure shows the set of line segments
that replace $e$ in the cut $\Gamma'$, such that $P(e)$ is
circumvented.}
\end{figure}

\end{document}